\documentclass[] {article}

\pdfoutput=1

\usepackage{authblk}
\usepackage{booktabs}
\usepackage{subcaption}
\usepackage{float}
\usepackage{amsmath}
\usepackage{amssymb}
\usepackage{amsthm}
\usepackage{dsfont}
\usepackage{paralist}
\usepackage[multiple]{footmisc}
\usepackage[ruled,vlined]{algorithm2e}
\usepackage{enumerate}
\usepackage{graphicx}
\usepackage[outdir=./]{epstopdf}
\usepackage{caption}
\usepackage{subcaption}
\usepackage{color}

\DeclareMathOperator*{\cov}{Cov} 
\DeclareMathOperator*{\var}{Var}	
\DeclareMathOperator*{\bias}{Bias}	
\DeclareMathOperator*{\mse}{MSE}	
\DeclareMathOperator*{\uniform}{Uniform}	
\DeclareMathOperator*{\FPPest}{\hat{f}_{FPP}} 

\newtheorem{theorem}{Theorem}
\newtheorem*{theorem*}{Theorem}

\newcommand{\id}{d_{i}} 
\newcommand{\od}{d_{o}} 
\newcommand{\avgdegree}{\bar{d}} 
\newcommand{\PCCodf}{\rho_{d_{o},f}} 
\newcommand{\SDf}{\sigma_{f}} 
\newcommand{\SDod}{\sigma_{d_{o}}} 
\newcommand{\Bglobal}{B_{global}} 
\newcommand{\Blocal}{B_{local}} 
\newcommand{\Blocalof}[1]{B_{local_{#1}}} 

\title{Friendship Paradox Skews Perceptions of Popularity in \\ Directed Networks}
\title{Friendship Paradox Distorts Observations in \\ Directed Networks}
\title{Friendship Paradox Biases Perceptions in Directed Networks}
\author[1]{Nazanin Alipourfard$^*$}
\author[2]{Buddhika Nettasinghe$^*$}
\author[1]{Andr\'es Abeliuk}
\author[2]{Vikram Krishnamurthy}
\author[1]{Kristina Lerman}
\affil[1]{Information Sciences Institute}
\affil[2]{Cornell Tech}

\begin{document}
\makeatletter
\def\blfootnote{\gdef\@thefnmark{}\@footnotetext}
\makeatother

\blfootnote{* N. Alipourfard and B. Nettasinghe contributed equally to this work.}
\maketitle

\begin{abstract}
How popular a topic or an opinion \emph{appears} to be in a network can be very different from its \emph{actual} popularity.  For example, in an online network of a social media platform, the number of people who mention a topic in their posts---i.e., its global popularity---can be dramatically different from how people see it in their social feeds---i.e., its perceived popularity---where the feeds aggregate their friends' posts. We trace the origin of this discrepancy to the friendship paradox in directed networks, which states that people are less popular than their friends (or followers) are, on average. We identify conditions on network structure that give rise to this perception bias, and validate the findings empirically using data from Twitter. Within messages posted by Twitter users in our sample, we identify topics that appear more frequently within the users' social feeds, than they do globally, i.e., among all posts. In addition, we present a polling algorithm that leverages the friendship paradox to obtain a statistically efficient estimate of a topic's global prevalence from biased perceptions of individuals. We characterize the bias of the polling estimate, provide an upper bound for its variance, and validate the algorithm's efficiency through synthetic polling experiments on our Twitter data. Our paper elucidates the non-intuitive ways in which the structure of directed networks can distort social perceptions and resulting behaviors.
\end{abstract}

\section{Introduction}

We observe our peers to learn social norms, assess risk, or copy behaviors. However, these observations can be systematically biased~\cite{Miller1994collective,Baer91,Prentice1993pluralistic,Kitts2003egocentric,berkowitz2005overview}, distorting how we see the world. One of the better known sources of bias is the friendship paradox in social networks~\cite{Feld91}, which states that people are less popular than their friends are, on average. 
Consequences of friendship paradox can 
skew how we compare ourselves to friends:
people tend to be less happy than their friends are~\cite{Bollen11}, and researchers tend to have less impact than their co-authors do, on average~\cite{Benevenuto2016}. In fact, any trait that is correlated with popularity is likely to be misperceived~\cite{Lerman2016majority,Eom14}. This may explain why adolescents systematically overestimate how much their peers drink or engage in risky behaviors~\cite{Baer91,berkowitz2005overview} and why social media use is often associated with negative social comparisons~\cite{abel2016social}.

In contrast to friendship networks, many online social networks are directed. On Twitter, for example, we subscribe to, or follow, others to see their posts, but the information does not flow in the opposite direction, unless those people also follow us back. For convenience, we refer to people whose posts we see our \emph{friends}, and those who see our posts our \emph{followers}. Note that this nomenclature does not imply a bidirectional friendship relationship. An individual's in-degree is the number of his or her friends, and the out-degree is the number of followers. The asymmetric nature of links in directed networks leads to four variants of the friendship paradox~\cite{Hodas13icwsm}: your friends (or followers) have more friends (or followers) than you do, on average. Empirically, this effect can be quite large, with upwards of 90\% of social media users observing that they have a lower in-degree and out-degree than both their friends and followers~\cite{Kooti14icwsm}. However, the conditions under which these four variants of the paradox exist have not been comprehensively analyzed. We carry out the analysis to show that while two variants of the friendship paradox occur in any directed network~\cite{higham2018centrality}, the remaining two exist only if an individual's in-degree and out-degree are correlated.

Friendship paradox can alter individual's observations of the network's state.
We consider directed networks where nodes have a trait, such as gender, political affiliation, or whether they used a certain hashtag in their posts. The trait's \emph{global prevalence} is simply the fraction of all nodes with that trait. On the other hand, its \emph{observed} prevalence is the 
fraction of friends of any node that have the trait.
In networks where the more influential (higher out-degree) nodes are likely to have the trait, its observed prevalence will be substantially higher than its actual prevalence. Our analysis shows that, similar to the generalized friendship paradox in undirected networks~\cite{Eom14_2,Eom14}, correlation between nodes' trait and their out-degree amplifies this perception bias.

In reality, an individual's perception of a trait is shaped by its \emph{local prevalence} among his or her friends. 
We identify a new paradox in directed networks, as a result of which
a trait will appear significantly more prevalent {locally} among individual's friends, than it is {globally} among all people. We show that this effect is stronger in networks where 
higher out-degree nodes are connected to nodes with a lower in-degree.

Surprisingly, although individual observations are biased, we can still make efficient estimates of the global prevalence of a trait.
We present a polling algorithm that obtains a statistically efficient estimate of a trait's global prevalence, with a smaller 
error than alternative polling methods. Proposed method 
leverages the friendship paradox 
to reduce the error of the polling estimate by trading off the bias of the estimate and its variance.
We analytically characterize this tradeoff and provide an upper bound for the variance.

We demonstrate that perception bias can be large in a real-world network.
To this end, we extracted a subgraph of the directed Twitter social network  and  collected messages posted by users within this subgraph. 
Treating the occurrence of particular hashtags within messages as traits or topics enables us to measure the perception bias. We identify hashtags that appear 
much more frequently within users' social feeds than they do among all messages posted by everyone, leading users to overestimate their prevalence.
We also validate the performance of the proposed polling algorithm through synthetic polling experiments on our Twitter subgraph.

Our paper elucidates some of the non-intuitive ways that directed social networks can bias individual perceptions. Since collective phenomena in networks, such as social contagion and adoption of social norms, are driven by individual perceptions, the structure of networks and the paradoxes endemic in them can impact social dynamics in unexpected way. This work shows how we can begin to quantify and mitigate these biases.

\section{Results}
\label{sec:results}

Consider a directed network $G = (V,E)$, with $\{V\}$ nodes and $\{E\}$ links. A link $(i,j)$ pointing from $i$ to $j$ indicates that $i$ is a \textit{friend} of $j$ or equivalently, $j$ \textit{follows} $i$. Here, the direction of the link indicates the flow of information. The out-degree of a node $v$, $\od(v)$, measures the number of followers it has, and its in-degree, $\id(v)$, the number of friends.

We define three random variables, $X$, $Y$ and $Z$, that correspond to different node sampling methods. A node $v$ with an out-degree $\od(v)$ has that many followers, or equivalently, $v$ is a friend to $\od(v)$ number of nodes. Therefore, a node $Y$ that is obtained from $V$ by sampling proportional to out-degree of nodes is called a \emph{random friend}. Similarly, a node $v$ that has $\id(v)$  links pointing to it is a follower of $\id(v)$ other nodes. Therefore, a node $Z$ that is obtained from $V$ by sampling proportional to in-degree of nodes is called a \emph{random follower}. Below, we formalize these terms.
	
	\begin{itemize}
		\item[\bf Random node $X$] is a uniformly sampled node from $V$:
		\begin{equation}
				\label{eq:distribution_X}
			\mathbb{P}(X = v) = \frac{1}{N} \quad \forall v\in V.
		\end{equation}

		\item [\bf Random friend $Y$] is a node sampled from $V$ proportional to its out-degree:
		\begin{equation}
		\label{eq:distribution_Y}
			\mathbb{P}(Y = v) = \frac{\od(v)}{\sum_{v'\in V}\od (v')}, \quad \forall v\in V.
		\end{equation}

		\item [\bf Random follower $Z$] is a node sampled from $V$ proportional to its in-degree:
		\begin{equation}
				\label{eq:distribution_Z}
			\mathbb{P}(Z = v) = \frac{\id(v)}{\sum_{v'\in V}\id (v')}, \quad \forall v\in V.
		\end{equation}
	\end{itemize}

\noindent For any directed network, the average in-degree $\mathbb{E}\{\id(X)\} = \frac{\sum_{v\in V} \id(v)}{N}$ and the average out-degree $\mathbb{E}\{\od(X)\} = \frac{\sum_{v\in V} \od(v)}{N}$ are the same.
Therefore, we use $\avgdegree$ to denote both average in-degree  and average out-degree  of a random node $X$: $\avgdegree=\mathbb{E}\{\od(X)\}=\mathbb{E}\{\id(X)\}$.

\subsection{Four Variants of the Friendship Paradox in Directed Networks}
	\label{subsec:four_versions_of_DFP}

Four different variants of the friendship paradox exist in directed networks~\cite{Hodas13icwsm}. The first two (Theorem \ref{th:friendship_paradox_any_network}) state that  (1) random friends have more followers than random nodes do, and (2) random followers have more friends than random nodes do (on average). The magnitudes of these are set by the heterogeneity (measured by the variance) of the in- and out-degree distributions of the underlying network. Theorems \ref{th:friendship_paradox_any_network} and \ref{th:friendship_paradox_positivelyCorrelated_inout_degree_network} were independently proved recently in \cite{higham2018centrality} utilizing vector norms. All omitted proofs are in the Appendix. 
	
	\begin{theorem}
		\label{th:friendship_paradox_any_network}
		Let $G = (V,E)$ be a directed network. Then,
		\begin{compactenum}
			\item random friend $Y$ has more followers than a random node $X$, on average; i.e.,
			\begin{equation}
\label{eq:dfpout}
				\mathbb{E}\{\od(Y)\} - \avgdegree = \frac{\var\{\od(X)\}}{\avgdegree} \geq 0.
			\end{equation}
			
			\item random follower $Z$ has more friends than a random node $X$, on average; i.e.,
			\begin{equation}
\label{eq:dfpin}
				\mathbb{E}\{\id(Z)\} - \avgdegree = \frac{\var\{\id(X)\}}{\avgdegree} \geq 0.
			\end{equation}
		\end{compactenum}
	\end{theorem}

The remaining two variants of the friendship paradox state that (3)~random friends have more friends than random nodes do, and (4)~random followers have more followers than random nodes do~(on average). In contrast to the first two variants of paradox (Theorem \ref{th:friendship_paradox_any_network}), these require positive correlation between the in-degree and out-degree of nodes in the network (Theorem \ref{th:friendship_paradox_positivelyCorrelated_inout_degree_network}).
	\begin{theorem}
		\label{th:friendship_paradox_positivelyCorrelated_inout_degree_network}
		Let $G = (V,E)$ be a directed network where in-degree $\id(X)$ and out-degree $\od(X)$ of a random node $X$ are positively correlated. Then,
		\begin{compactenum}
			\item random friend $Y$ has more friends than a random node $X$ does, on average; i.e.,
			\begin{equation}
				\mathbb{E}\{\id(Y)\} - \avgdegree = \frac{\cov\{\id(X),\od(X)\}}{\avgdegree} \geq 0.
			\end{equation}
			
			\item random follower $Z$ has more followers than a random node $X$ does, on average; i.e.,
			\begin{equation}
				\mathbb{E}\{\od(Z)\} - \avgdegree = \frac{\cov\{\id(X),\od(X)\}}{\avgdegree} \geq 0.
			\end{equation}
		\end{compactenum}
	\end{theorem}

	Theorem \ref{th:friendship_paradox_positivelyCorrelated_inout_degree_network} states that in networks where the in- and out-degrees of a random node are positively correlated, (1) the expected number of friends of a random friend is greater than the expected number of friends of a random node, and (2) the expected number of followers of a random follower is greater than that of a random node.
	
\begin{figure}[H]
    \centering
    \begin{subfigure}[t]{0.48\columnwidth}
        \centering
        \includegraphics[width=\columnwidth]{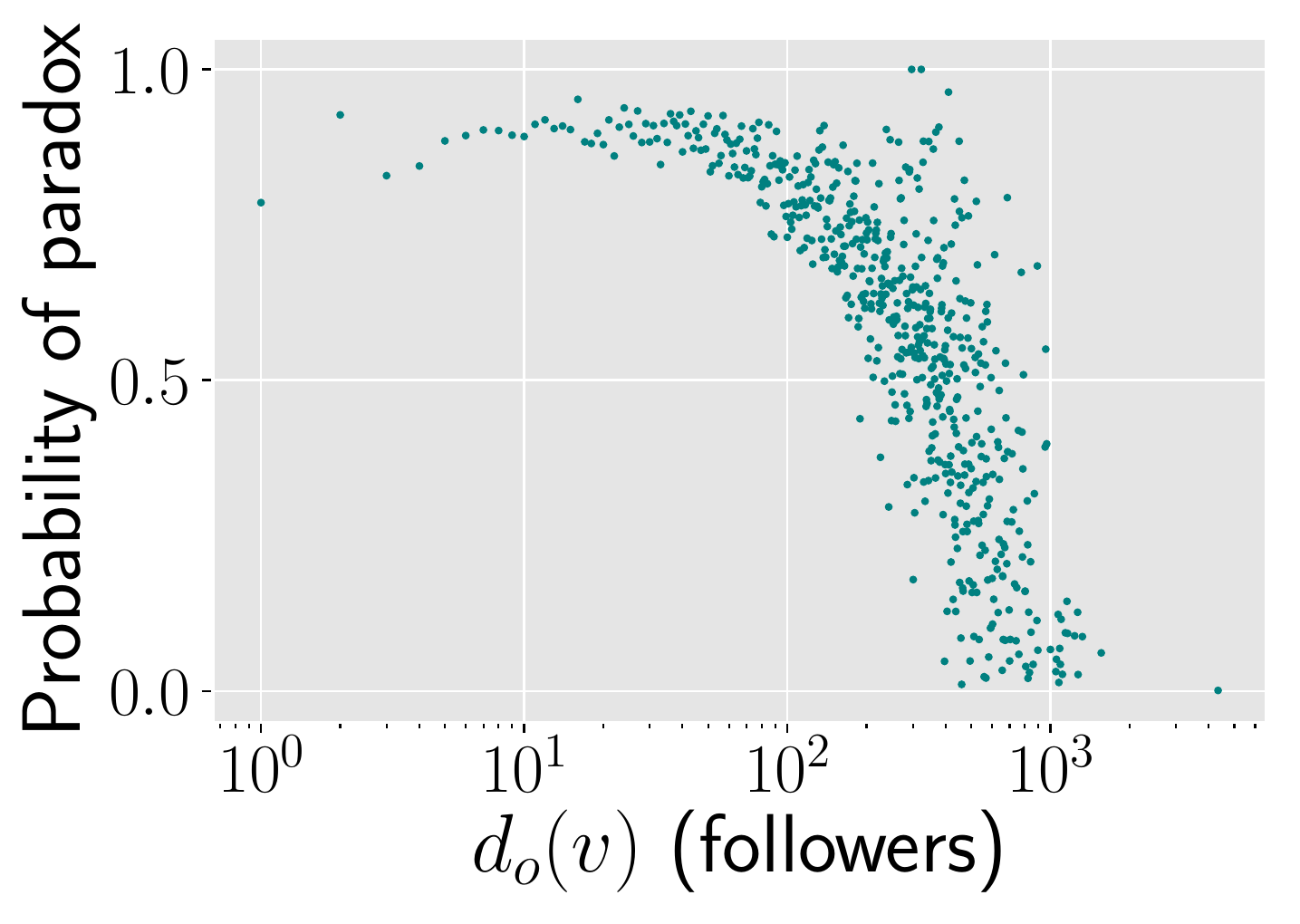}
        \caption{Friends have more followers}\label{fig:sp-oo}
    \end{subfigure}
    \begin{subfigure}[t]{0.48\columnwidth}
        \centering
        \includegraphics[width=\columnwidth]{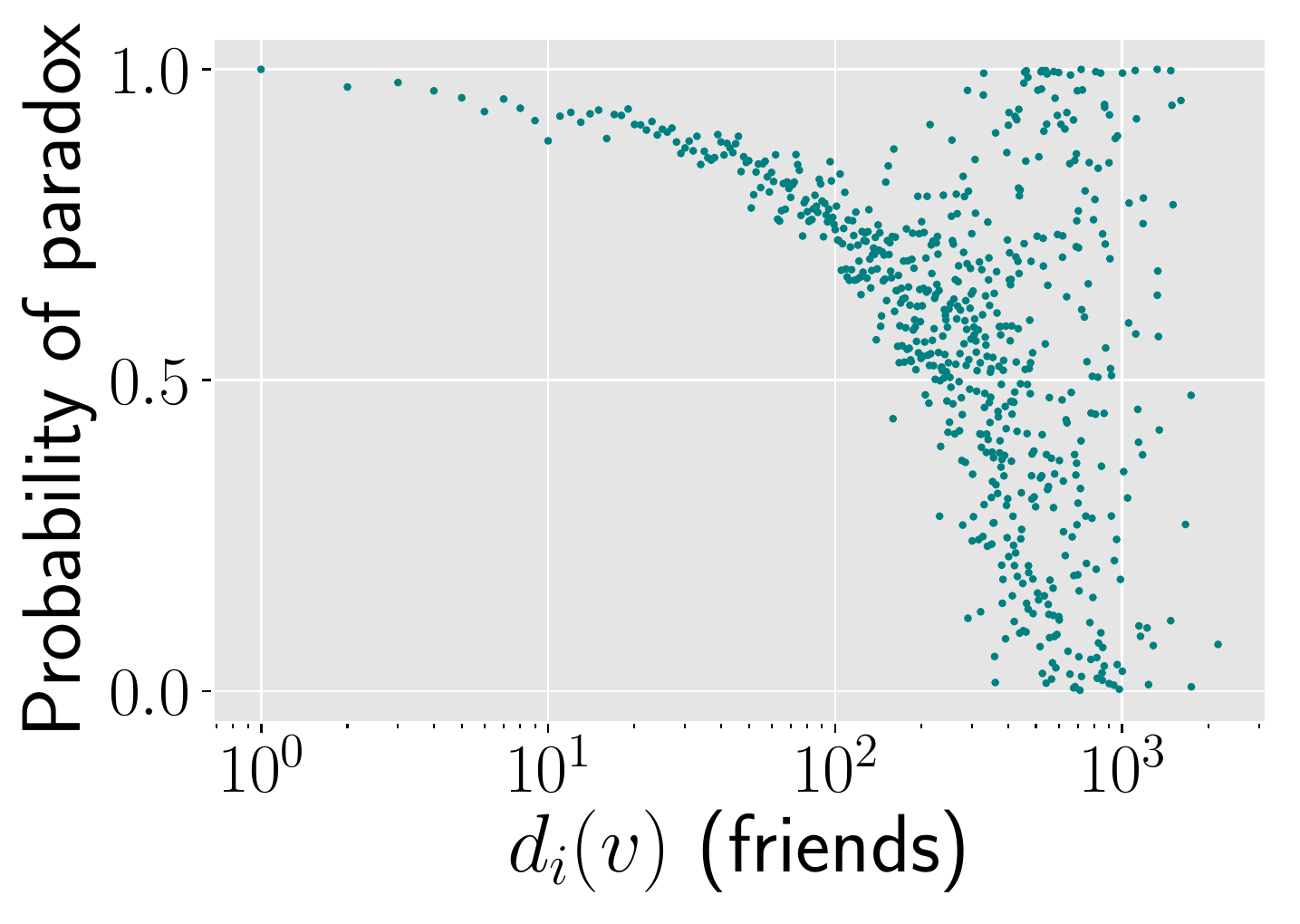}
        \caption{Followers have more friends}\label{fig:sp-ii}
    \end{subfigure}
       \begin{subfigure}[t]{0.48\columnwidth}
        \centering
        \includegraphics[width=\columnwidth]{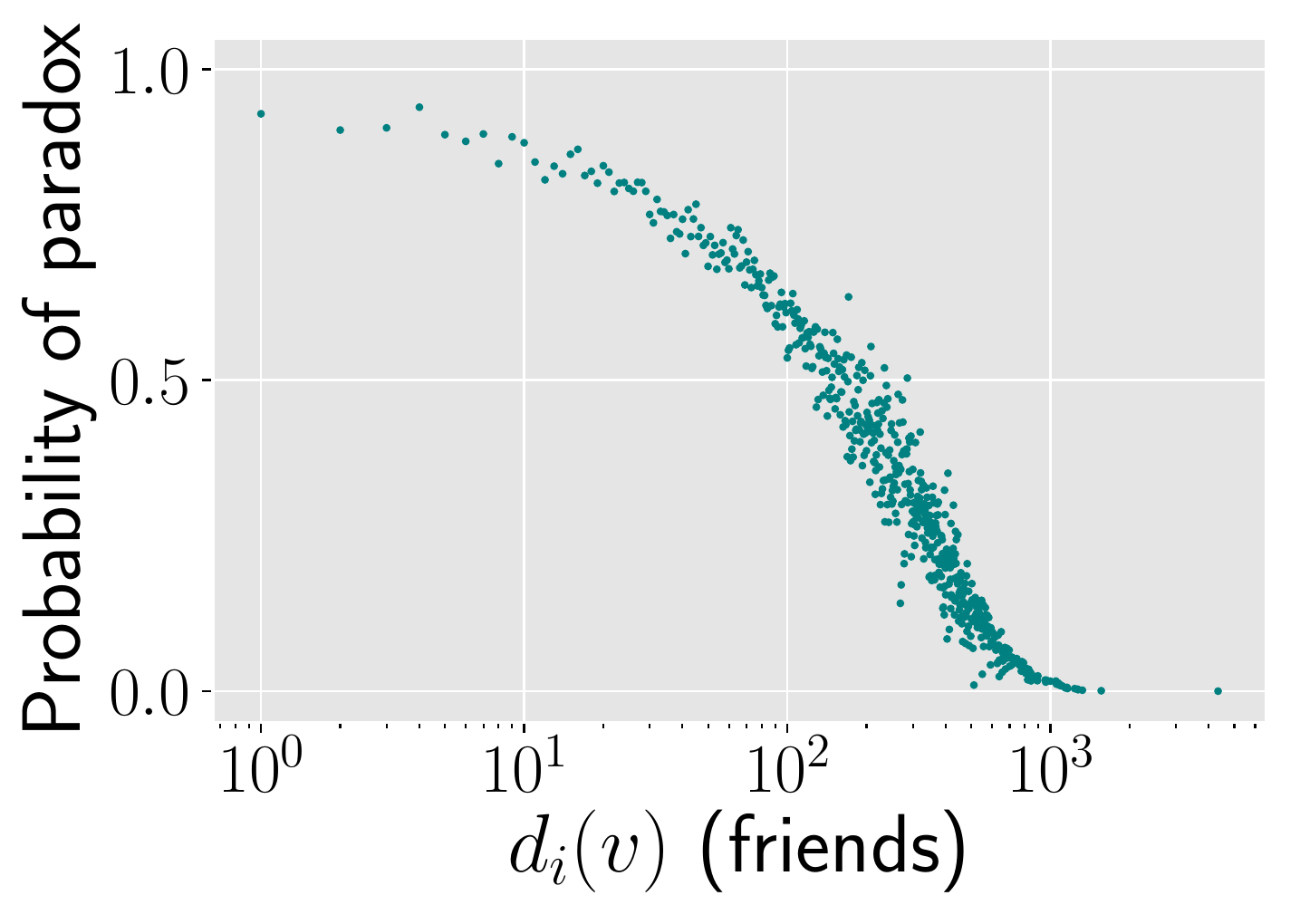}
        \caption{Friends have more friends}\label{fig:sp-io}
    \end{subfigure}
    \begin{subfigure}[t]{0.48\columnwidth}
        \centering
        \includegraphics[width=\columnwidth]{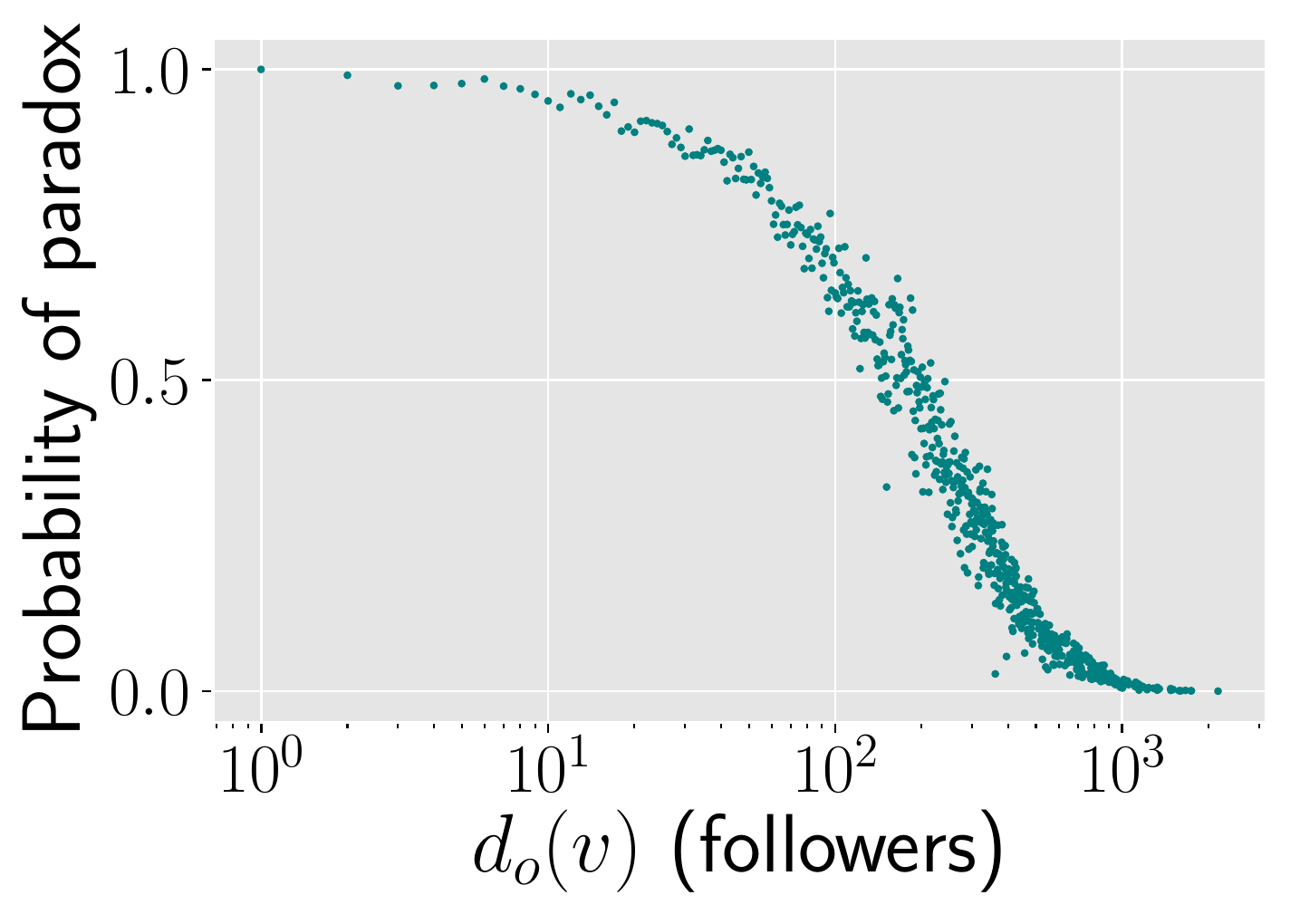}
        \caption{Followers have more followers}\label{fig:sp-oi}
    \end{subfigure}
    \caption{Illustration of the effects of the four versions of the friendship paradox using dataset described in Sec. \ref{subsecc:data}.
    The sub-figures display the fraction of nodes (empirical probability of the paradox) of a particular degree whose (a) friends have more followers, (b) followers have more friends, (c) friends have more friends, and (d) followers have more followers, on average.
    \label{fig:sp}}
\end{figure}

Figure \ref{fig:sp} illustrates the four variants of the friendship paradox in the subgraph of the directed social network of Twitter (see Methods). Specifically, it shows the fraction of individuals with a specific in-degree (or out-degree) that experiences the paradox. Note that this fraction is high: at least half of the users with fewer than 100 friends or followers observe that they are less popular and well-connected than their friends and followers are, on average.

\subsection{Perception Biases in Directed Networks}
	\label{subsec: perception_bias_and_FP}
	
When nodes have distinguishing traits or attributes, friendship paradox can bias perceptions of those attributes. For simplicity we assume that each node has a binary valued attribute  ($f: V \rightarrow \{0,1\}$). Such binary functions are useful for representing, among others, voting preferences (Democratic or Republican), demographic characteristics (female or male), contagions (infected vs susceptible), or the spread of information in networks (using a particular hashtag or not).

\subsubsection{Global Perception Bias}	
The \emph{global prevalence} of the attribute in a directed network is given by $\mathbb{E}\{f(X)\}$, the expected attribute value of a random node $X$. In other words, when only 5\% of nodes have the attribute $f(v)=1$, for example, they tweeted about a topic, its expected value is  $\mathbb{E}\{f(X)\}=0.05$.
However, nodes' perceptions of the prevalence of the attribute are determined by its value among their friends.  In other words, 
nodes' perception of how prevalent the attribute is, is given by the expected attribute value of a randomly chosen friend $Y$: $\mathbb{E}\{f(Y)\}$. On Twitter this translates into how many people see the topic in their social feed, which aggregates posts made by friends.
Under some conditions, the perceived prevalence of the attribute $\mathbb{E}\{f(Y)\}$ will be very different from its actual prevalence $\mathbb{E}\{f(X)\}$.  We define this as \emph{global perception bias}:
	
		\begin{align}
			\Bglobal = \mathbb{E}\{f(Y)\} - \mathbb{E}\{f(X)\} &= \frac{\cov(f(X),\od(X))}{\avgdegree}\\
			&=\frac{\PCCodf\SDod \SDf}{\avgdegree},	\label{eq:global_perception_bias_Y}
		\end{align}

\noindent where $\PCCodf$ is the Pearson correlation coefficient between out-degree and attribute value of a random node, $\SDod$ is the standard deviation of the out-degree distribution, and $\SDf$ is the standard deviation of the binary attributes (see appendix for the derivation).

When the attribute is correlated with out-degree ($\PCCodf > 0$), a random friend's attribute is larger than the attribute value of a random node, on average. In undirected networks this effect is known as \emph{generalized friendship paradox}~\cite{Eom14}, and it has the same intuition: 
when popular people (i.e., those with many followers) are more likely to possess some trait ($\PCCodf > 0$), that trait will be overrepresented among the friends of any individual. As a result, people will tend to overestimate the trait's prevalence. This may explain the observation that adolescents overestimate the number of smokers or heavy drinkers among their peers~\cite{Baer91}. All that is required for the bias to hold is if peers with risky behaviors tended to be more popular.
	
Note that 
the magnitude of the friendship paradox $S_{FP} = \mathbb{E}\{\od(Y)\} - \avgdegree = \frac{\sigma^2_{\od}}{\avgdegree}$ increases with the standard deviation of the out-degree distribution ($\sigma_{\od }$) and decreases with the average degree ($\avgdegree$). Global perception bias $\Bglobal$ also increases with $\sigma_{\od }$ and decreases with $\avgdegree$ when the correlation coefficient $\PCCodf$ remains fixed. Hence, friendship paradox amplifies global perception bias, increasing the deviation between the actual and observed prevalence 
of the attribute in the network.

\subsubsection{Local Perception Bias}
Since information in a directed network flows to individuals from their friends, their perceptions of the world are given by the 
values of the attribute among their friends. One problem with using $\Bglobal$ to measure perception bias is that $\mathbb{E}\{f(Y)\}$ captures the expected attribute value among the friends of 
all individuals, rather than friends 
of a randomly chosen individual $X$.
Therefore, we define an alternate measure of perception bias---\emph{local perception bias}---
that considers a node's perception of an attribute based on its expected value among its friends.

Formally, the perception $q_f(v)$ of a node $v\in V$  about the 
prevalence of an attribute $f$ is
	\begin{equation}
	\label{eq:defn_qf}
		q_f(v) = \frac{\sum_{ u \in {Fr}(v)}f(u)}{\id (v)},
	\end{equation}
where $Fr(v)$ denotes the set of friends of $v$. Local perception bias is then the deviation of the expected perception of a random individual from its global prevalence:
		\begin{align}
			\Blocal = {\mathbb{E}\{q_f(X)\}} - {\mathbb{E}\{f(X)\}}.	\label{eq:local_perception_bias_Y}
		\end{align}
To help quantify this value, we define attention that a node $v\in V$ allocates to each of her friends:
	\begin{equation}
		\mathcal{A}(v) = \frac{1}{\id (v)}.
	\end{equation}
The analogy is motivated by an observation that users with more friends tend to receive more messages~\cite{GomezRodriguez13}, making them less likely to see any specific friend's post~\cite{Hodas12socialcom}. This allows us to succinctly express the expected perception of a random node $X$ as  $\mathbb{E}\{q_f(X)\} = \avgdegree \cdot \mathbb{E}\{f(U) \mathcal{A}(V)\vert (U, V) \sim \uniform(E) \}$  (see appendix for derivation). Here, $\avgdegree$ is the expected number of friends of a random node, and $U$ and $V$ denote the endpoints of a link sampled uniformly from $E$. Intuitively, $\mathbb{E}\{f(U) \mathcal{A}(V)\vert (U, V) \sim \uniform(E) \}$ represents the expected influence of an interaction along an link drawn at random from the network: i.e., the attribute $f(U )$ of the friend $U$ times the attention that the follower $V$ pays to that friend. Hence, the expected perception $\mathbb{E}\{q_f(X)\}$ of a random node $X$ is the product of the average number of interactions $\avgdegree$ and the average influence of an interaction.

The appendix notes show that local perception bias exists, i.e., ${\Blocal \geq 0}$, which indicates the local overestimation of global prevalence, if the following conditions are met:
		\begin{align}
			\label{eq:suff_condition_1_local_perception_bias}
			\cov\{f(X), \od(X)\} &\geq 0 \quad \text{and,}\\
			\label{eq:suff_condition_2_local_perception_bias}
			\cov\{f(U), \mathcal{A}(V)\vert {(U,V) \sim \uniform(E)}\} &\geq 0.
		\end{align}
The first condition (Eq.~\ref{eq:suff_condition_1_local_perception_bias})
specifies positive correlation between the out-degree and the attribute of a random node,
 implying that popular nodes are more likely to have the attribute. The second condition (Eq.~\ref{eq:suff_condition_2_local_perception_bias}) 
specifies positive correlation between the attention of a follower and the attribute of a friend, 
suggesting that nodes with higher attribute values will appear as friends of nodes that follow few others. 
These two conditions are sufficient 
for positive local perception bias ($\Blocal > 0$),  leading individuals to overestimate the attribute's prevalence.
Further, the two conditions~(Eq.~\ref{eq:suff_condition_1_local_perception_bias} and Eq.~\ref{eq:suff_condition_2_local_perception_bias}) also ensure that $\Blocal > \Bglobal > 0$ as shown in Appendix Section 2. 
Hence, under these two conditions, local perception bias and global perception bias will both indicate overestimation of the global prevalence.

However, $\Bglobal$ and $\Blocal$ can differ significantly in certain settings. For example, there exist situations where the two measures have different signs, with one measure suggesting overestimation and the other suggesting underestimation of an attribute's prevalence by individuals in the social network. 
In such cases, we propose using $\Blocal$ to measure the perception bias, as it takes more structural properties of the network into account. Further, global perception bias $\Bglobal$ and local perception bias $\Blocal$ are equal if and only if the attribute $f(U)$ of $U$ and attention $\mathcal{A}(V)$ of a random link $(U,V)$ are uncorrelated, i.e.,
		\begin{align}
		\cov\{f(U), \mathcal{A}(V)\vert {(U,V) \sim \uniform(E)}\} &= 0
		\end{align}
as we show in the appendix notes \ref{appendix:th_local_PB}.


\subsubsection{Empirical Validation}
To measure perception bias, we used data from Twitter (see Methods) to compare the actual and perceived popularity of various hashtags mentioned in text posts. We treat each hashtag $h$ as a binary attribute, with $f_h(v)=1$ if a user $v$ used the hashtag $h$ in his or her posts.

\begin{figure}[H]
    \centering
    \begin{subfigure}[t]{0.48\columnwidth}
        \centering
        \includegraphics[width=\columnwidth]{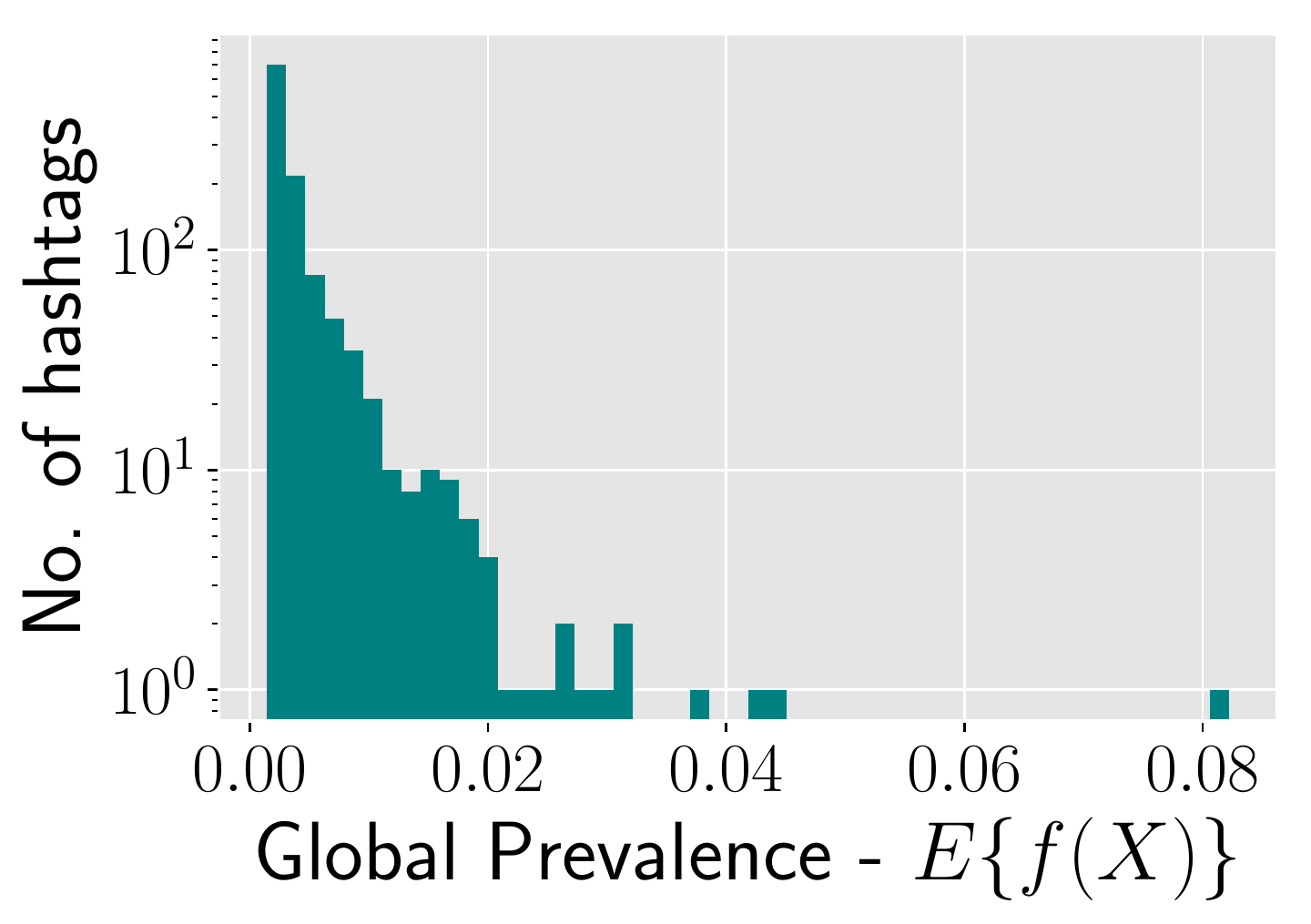}
        \caption{\label{fig:efx}}
    \end{subfigure}
    \begin{subfigure}[t]{0.48\columnwidth}
        \centering
        \includegraphics[width=\columnwidth]{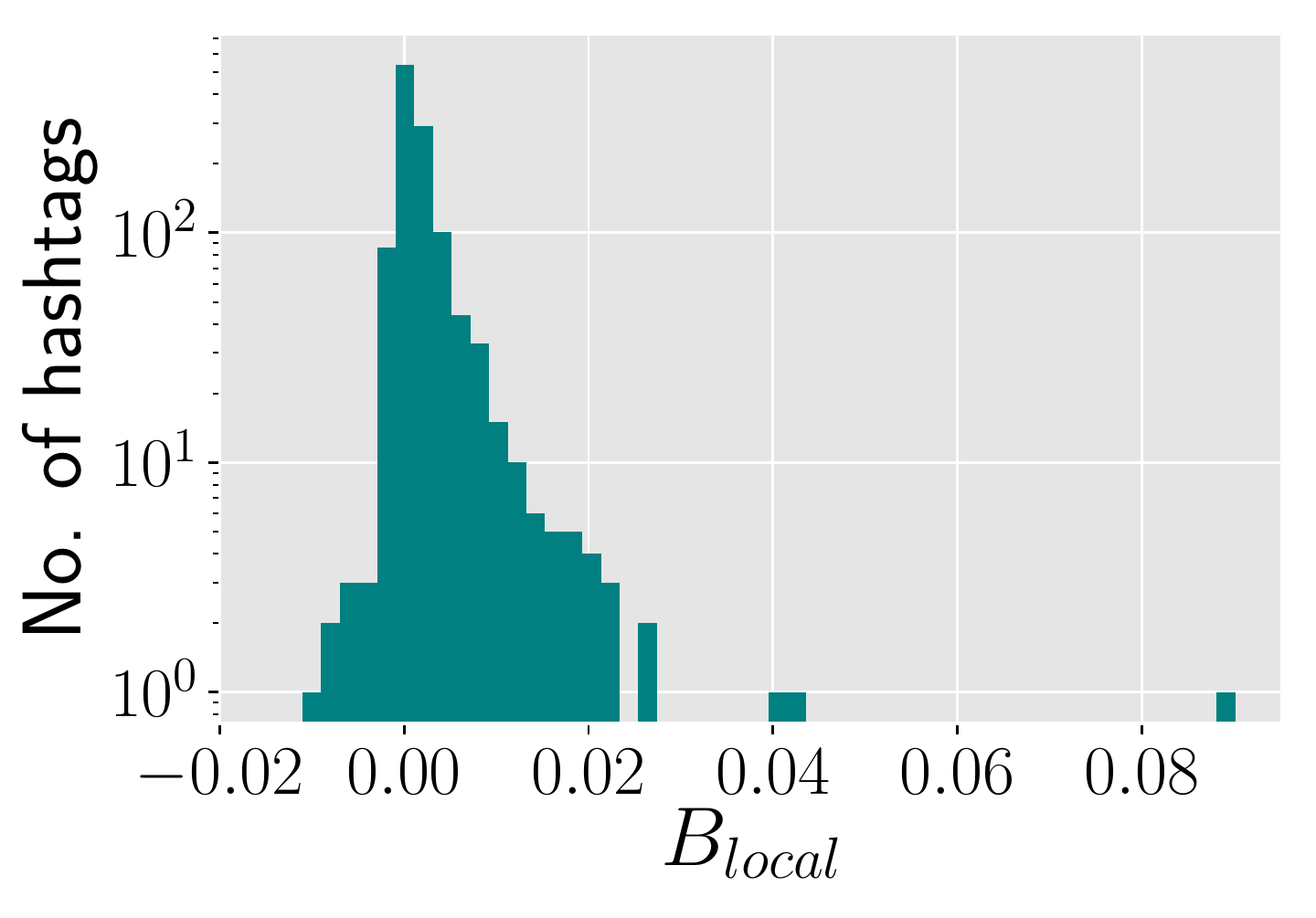}
        \caption{\label{fig:local}}
    \end{subfigure}
    \caption{
    	Histogram of the distribution of (a) global prevalence $\mathbb{E}\{f(X)\}$ and (b) local perception bias $\Blocal$ of popular hashtags in the Twitter data. Local perception bias $\Blocal$ (overestimating the prevalence) exists for most hashtags.  
    \label{fig:covs}}
\end{figure}

Figure~\ref{fig:efx} shows the histogram of the prevalence ($\mathbb{E}\{f(X)\}$) of the 1,153 most popular hashtags, each used by more than 1,000 people in our data set. The bulk of these hashtags were used by fewer than 2\% of the people, with the most popular hashtags being used by just 8\% of the people in the subgraph.
Figure~\ref{fig:local} shows the histogram of $\Blocal$ value for all hashtags. Although its peak is at zero, the distribution is skewed, with $865$ hashtags having positive bias, meaning that they appear more popular than they really are.

\begin{figure}
    \includegraphics[width=0.9\columnwidth]{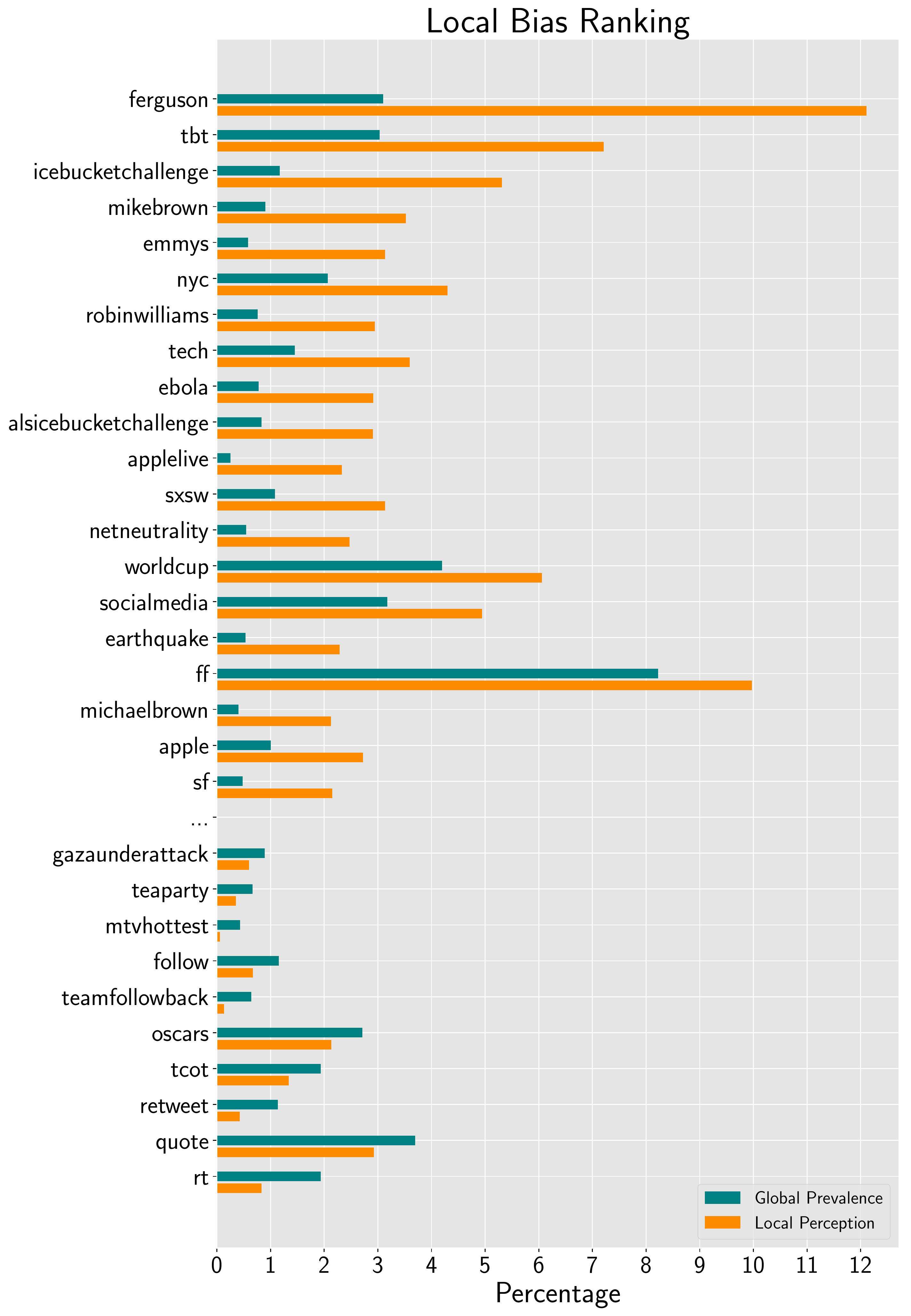}

 \resizebox{0.9\textwidth}{!}{\renewcommand{\arraystretch}{0.9}%
 \begin{tabular}{|@{ }c@{ }|@{ }c@{ }|@{ }c@{ }|@{ }c@{ }||@{ }c@{ }|@{ }c@{ }|@{ }c@{ }|@{ }c@{ }|}
 \hline
 \multicolumn{1}{|c|}{ff} & \multicolumn{7}{|c|}{\textbf{F}ollow \textbf{F}riday: introducing account worth following.} \\
 \multicolumn{1}{|c|}{rt} & \multicolumn{7}{|c|}{\textbf{R}e\textbf{T}weet} \\
 \multicolumn{1}{|c|}{tbt} & \multicolumn{7}{|c|}{\textbf{T}hrow\textbf{B}ack \textbf{T}hursday: posting an old picture on Thursdays.} \\
 \multicolumn{1}{|c|}{tcot} & \multicolumn{7}{|c|}{\textbf{T}op \textbf{C}onservatives \textbf{O}n \textbf{T}witter} \\
 \multicolumn{1}{|c|}{mike(/michael)brown, ferguson} & \multicolumn{7}{|c|}{an 18-years-old African American man killed by police.} \\
 \multicolumn{1}{|c|}{(als) icebucketchallenge} & \multicolumn{7}{|c|}{A challenge to promote awareness of ALS disease.} \\
 \multicolumn{1}{|c|}{sxsw} & \multicolumn{7}{|c|}{South by Southwest: Annual conglomerate of film and music festivals.} \\
 \hline
 \end{tabular}
 }
    \caption{The ranking of popular Twitter hashtags based on \textit{Local Bias}. Top-$20$ and bottom-$10$ are included in the ranking.  The bars compare $\mathbb{E}\{f(X)\}$ (global prevalence) and $\mathbb{E}\{q_f(X)\}$ (local perception). The hashtags can appear to be much more popular than they actually are (e.g. \textit{\#ferguson}) or, they can appear to be less popular (e.g. \textit{\#oscars}) due to local perception bias. \label{fig:local_ranking}}
\end{figure}

What hashtags have most bias? 
Figure~\ref{fig:local_ranking} shows the top-20 and bottom-10 hashtags ranked by $\Blocal$.
Among the most positively biased hashtags are those associated with social movements (\textit{\#ferguson, \#mikebrown, \#michaelbrown}), memes and current events (\textit{\#icebucketchallenge, \#alsicebucketchallenge}, \textit{\#ebola}, \textit{\#netneutrality}), sports and entertainment (\textit{\#emmys}, \textit{\#robinwilliams}, \textit{\#sxsw}, \textit{\#applelive}, \textit{\#worldcup}). For example, \textit{\#ferguson}, with $\mathbb{E}\{q_{f}(X)\}=12.1\%$, 
is perceived as the most popular hashtag. While it is also one of the more widely-used hashtags, with $\mathbb{E}\{f(X)\}=3.1\%$, perception bias makes it appear about four times more popular to Twitter users, on average, than it actually is. Interestingly, there are also hashtags with negative bias, indicating that they appear less popular than they actually are. Among these hashtags are Twitter conventions aimed at getting more followers (\textit{\#tfb}, \textit{\#followback}, \textit{\#follow}, \textit{\#teamfollowback}) or more retweets (\textit{\#shoutout}, \textit{\#pjnet}, \textit{\#retweet}, \textit{\#rt}).
Many of these hashtags are actually among the top-20 most popular Twitter hashtags (\textit{\#oscars}, \textit{\#tcot}, \textit{\#quote} and \textit{\#rt}), but due to the structure of the network, they appear less popular to users.
This occurs either because people who use these hashtags do not have many followers ($\cov\{f(X), \od(X)\}<0$), or the attention of their followers is diluted because they follow many others ($\cov\{f(U), \mathcal{A}(V)\}<0$). For example, for \textit{\#oscars}, both of the covariances are negative. The ranking of hashtags based on global bias is available in Figure \ref{fig:global_bias_ranking}.

\begin{figure}[H]
    \centering
    \begin{subfigure}[t]{0.48\columnwidth}
        \centering
        \includegraphics[width=\columnwidth]{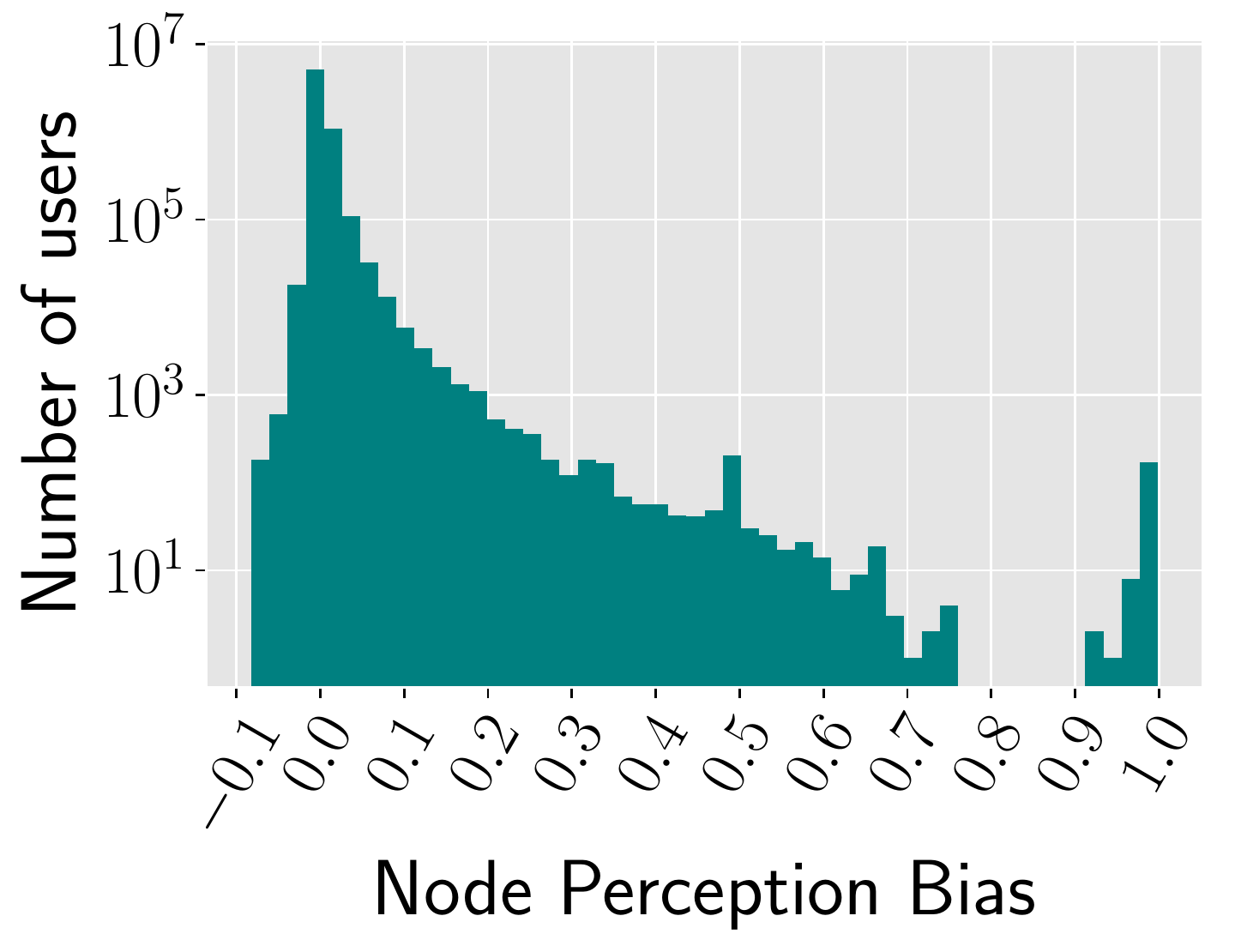}
        \caption{\label{fig:individual}}
    \end{subfigure}
    \begin{subfigure}[t]{0.48\columnwidth}
        \centering
        \includegraphics[width=\columnwidth]{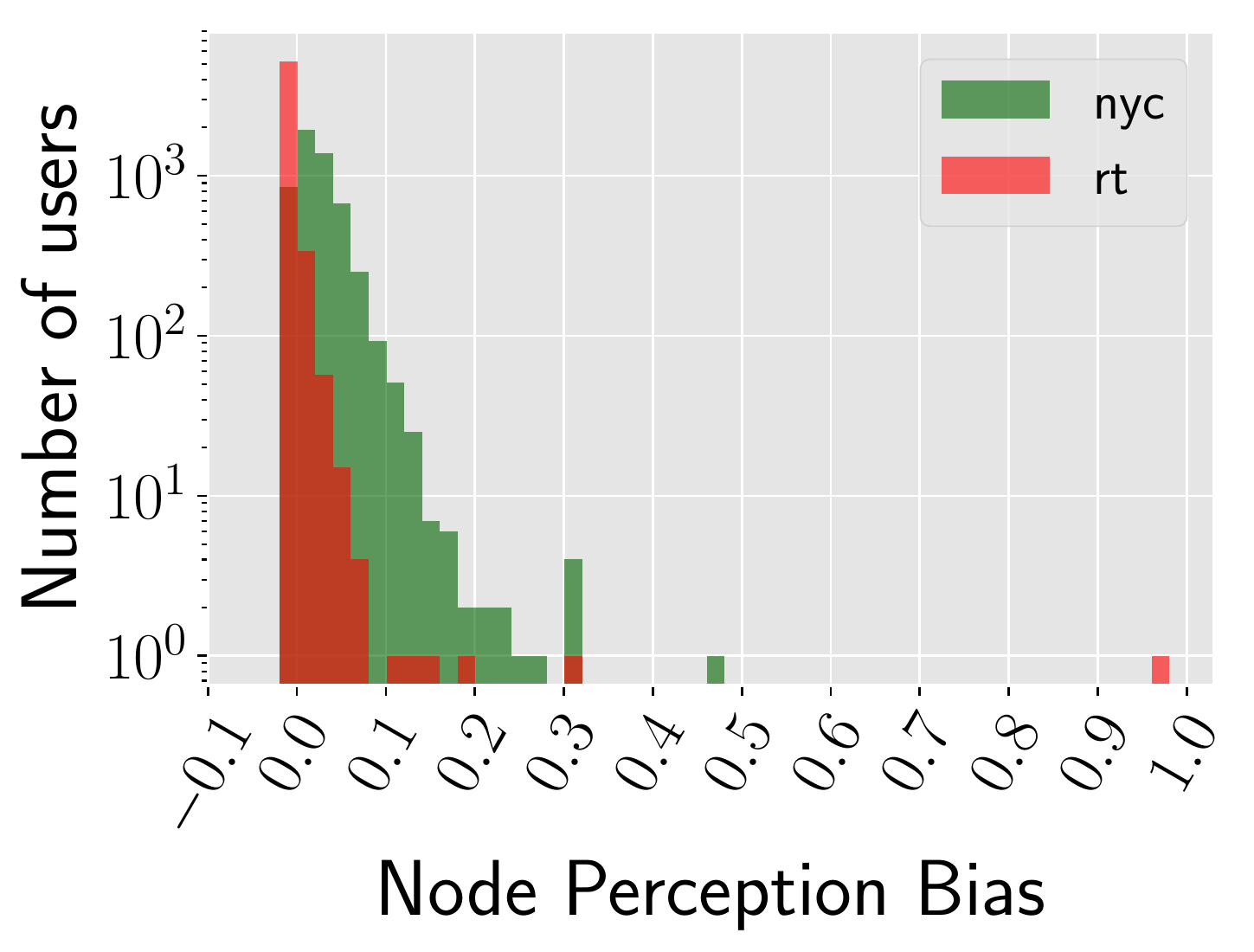}
        \caption{\label{fig:hashtag_individual}}
    \end{subfigure}
    \caption{ Individual-level perception bias  ${q_{f_{h}}(v)-\mathbb{E}\{f(X)\}}$ for (a) all hashtags $h$ and all nodes $v\in V$,  and (b) for two hashtags with similar global prevalence, but with positive (\textit{\#nyc}) and negative (\textit{\#rt}) $\Blocal$. This illustrates that most hashtags are positively biased for individuals, with bias levels that do not depend on global prevalence.  \label{fig:userlevel}}
\end{figure}

At an individual level, the popularity of a hashtag $h$ among the friends of a user $v\in V$ is given by $q_{f_{h}}(v)$. 
The individual-level perception bias is then $B_{h}(v)=q_{f_{h}}(v) - \mathbb{E}\{f_{h}(X)\}$, where $\mathbb{E}\{f_{h}(X)\}$ is the global prevalence of hashtag $h$. Figure~\ref{fig:userlevel} shows the empirical distribution of $B_{h}(v)$ for all users and hashtags. Most of the mass of the histogram is for $B_{h}(v)>0$, suggesting that most of the people in our data overestimate the popularity of these hashtags.

Figure \ref{fig:hashtag_individual} compares individual-level perception bias for two hashtags that have similar global prevalence:  \textit{\#nyc} ($\mathbb{E}\{f(X)\}= 0.021$) and \textit{\#rt} ($\mathbb{E}\{f(X)\}= 0.019$). Of the two hashtags, \textit{\#nyc} is perceived as more popular (with $\Blocalof{\textit{\#nyc}}= 0.022$), but
\textit{\#rt} appears less popular (with $\Blocalof{\textit{\#rt}}= -0.011$) than it is globally.

\subsection{Estimating Global Prevalence via Polling}
\label{sec:polling}



Polling estimates the global prevalence $\mathbb{E}\{f(X)\}$ of an attribute by sampling random individuals and averaging their answers to some question. The accuracy of a poll depends on two key factors: (i) the method of sampling individuals (sampling distribution) and, (ii) the question presented to them. We propose a practical polling algorithm (Algorithm \ref{alg:FPP}) that differs from the currently used polling algorithms in both aspects. First, our algorithm samples random followers (step 1 of Algorithm~\ref{alg:FPP}) instead of random individuals, as is done by most alternative methods. Second, instead of asking about their own attribute, the sampled individuals are asked about their perception (step 2 of Algorithm~\ref{alg:FPP}):
	
	{\centering\textit{``What do you think is fraction of individuals with attribute 1?"}}
	
\noindent	 Consequently, we call the proposed algorithm \textit{Follower Perception Polling (FPP)} algorithm.
	
	 \begin{algorithm}
	 	\caption{Follower Perception Polling (FPP) Algorithm}
	 	\label{alg:FPP}
	 	\DontPrintSemicolon 
	 	\KwIn{Graph $G = (V, E)$, perceptions $q_f:V\rightarrow \mathbb{R}^+$, sampling budget $b$.}
	 	\KwOut{Estimate $\FPPest$ of $\mathbb{E}\{f(X)\} = \frac{\sum_{v \in V}f(v)}{N}$.}
	 	
	 	\vspace{0.3cm}
	 	
	 	\begin{enumerate}
	 		\item Sample a set $S \subset V$ of $b$ followers independently from the distribution $$p_v =  \frac{\id(v)}{\sum_{v'\in V}\id (v')}, \quad \forall v\in V.$$

	 		\item Compute the estimate
	 		\begin{equation}
	 		\label{eq:NEP_estimate}
	 		\FPPest = \frac{1}{b} \sum_{ v \in S}  q_f(v).
	 		\end{equation}
	 	\end{enumerate}
	 \end{algorithm}

As random followers have more friends than random nodes (on average), according to Theorem \ref{th:friendship_paradox_any_network}, the key idea behind the FPP algorithm is to sample individuals who have more friends. As a result, the variance of the perceptions of random followers will be smaller (compared to that of random nodes) and hence, will result in a more accurate (lower mean-squared error) estimate of the global prevalence of the attribute. We analytically show that (see Methods) (i) the bias of the estimate $\FPPest$ produced by the FPP algorithm is same as the global perception bias $\Bglobal$ and, (ii) variance of the estimate $\FPPest$ produced by the FPP algorithm is bounded above by a function of the correlation between out-degree and the attribute as well as spectral properties of the network (i.e. second largest eigenvalue of the bibliographic coupling matrix).

\begin{figure}[tbh]
    \centering
    \begin{subfigure}[t]{0.49\columnwidth}
        \centering
        \includegraphics[width=\columnwidth]{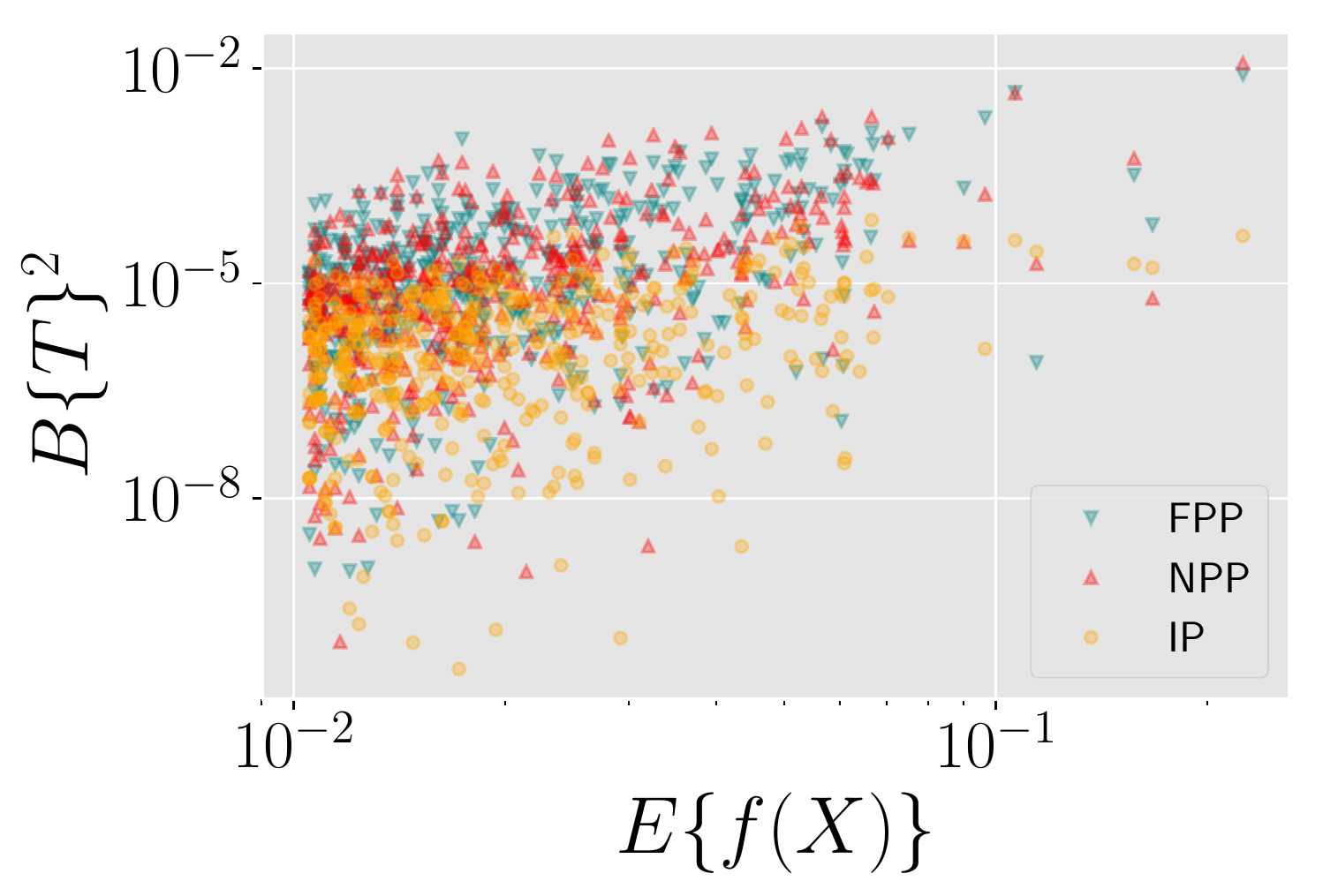}
        \caption{\label{fig:bias_polling}}
    \end{subfigure}
    \centering
    \begin{subfigure}[t]{0.49\columnwidth}
        \centering
        \includegraphics[width=\columnwidth]{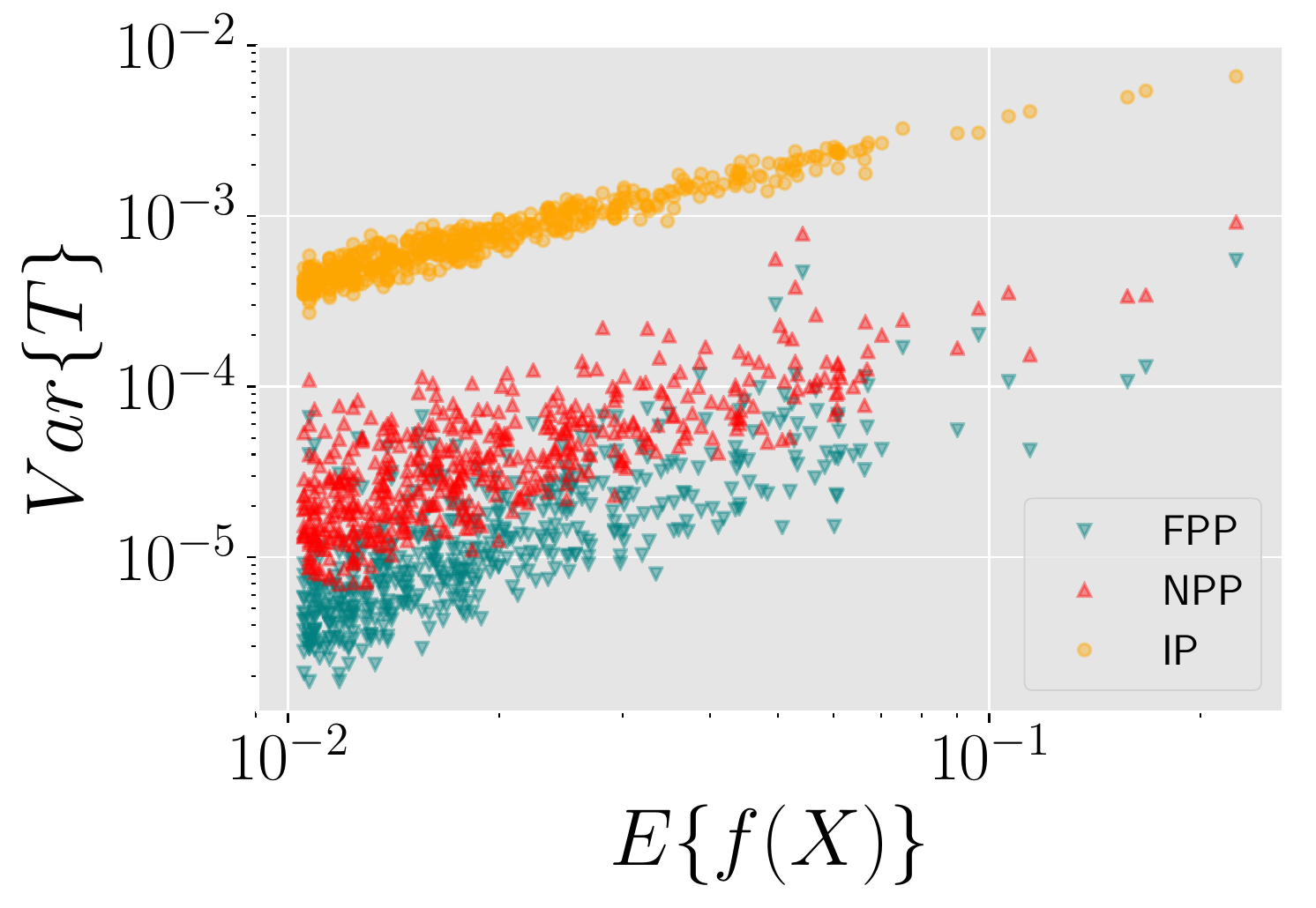}
        \caption{\label{fig:var_polling}}
    \end{subfigure}

    \begin{subfigure}[t]{0.49\columnwidth}
        \centering
        \includegraphics[width=\columnwidth]{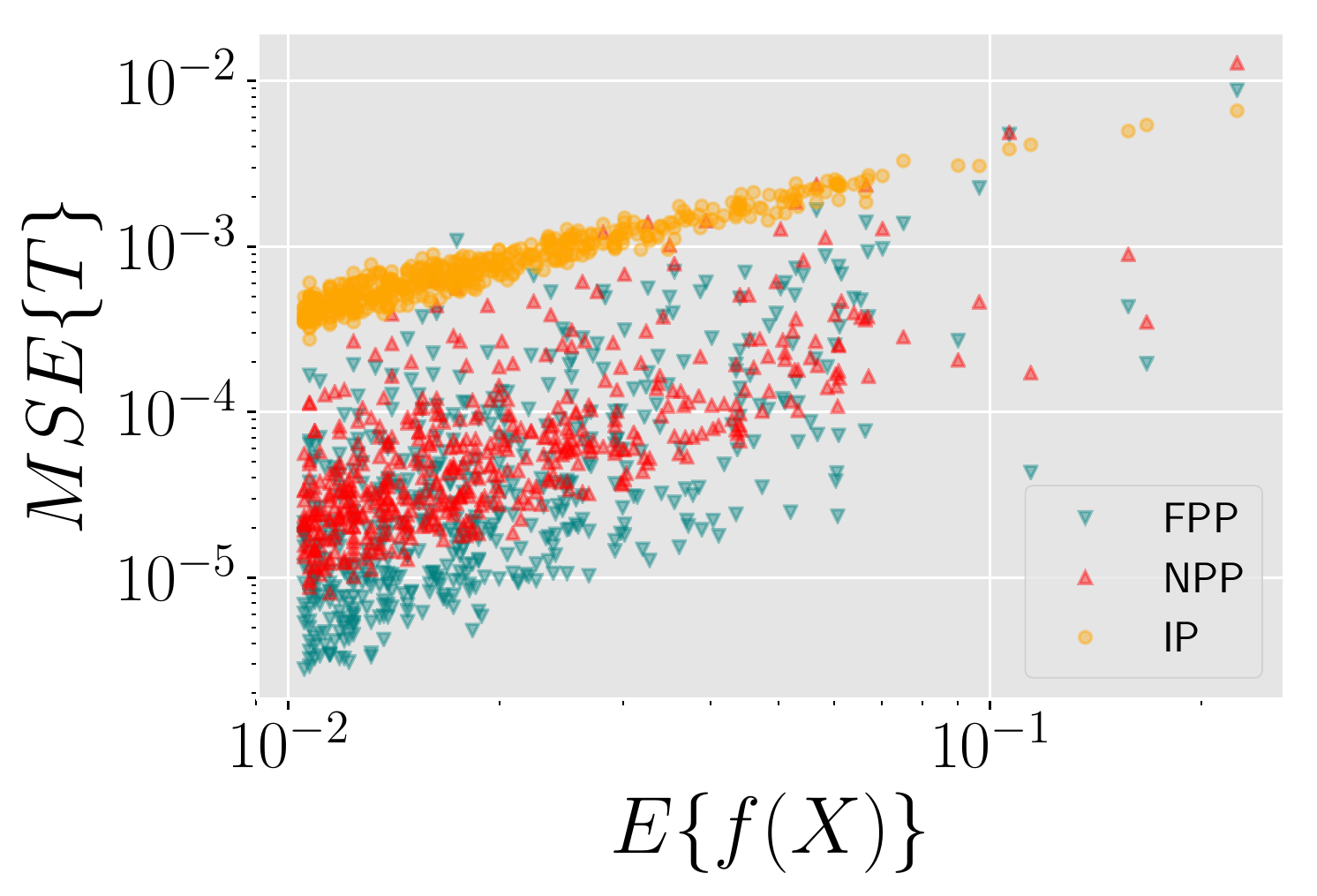}
        \caption{\label{fig:mse_polling}}
    \end{subfigure}
    \begin{subfigure}[t]{0.49\columnwidth}
        \centering
        \includegraphics[width=\columnwidth]{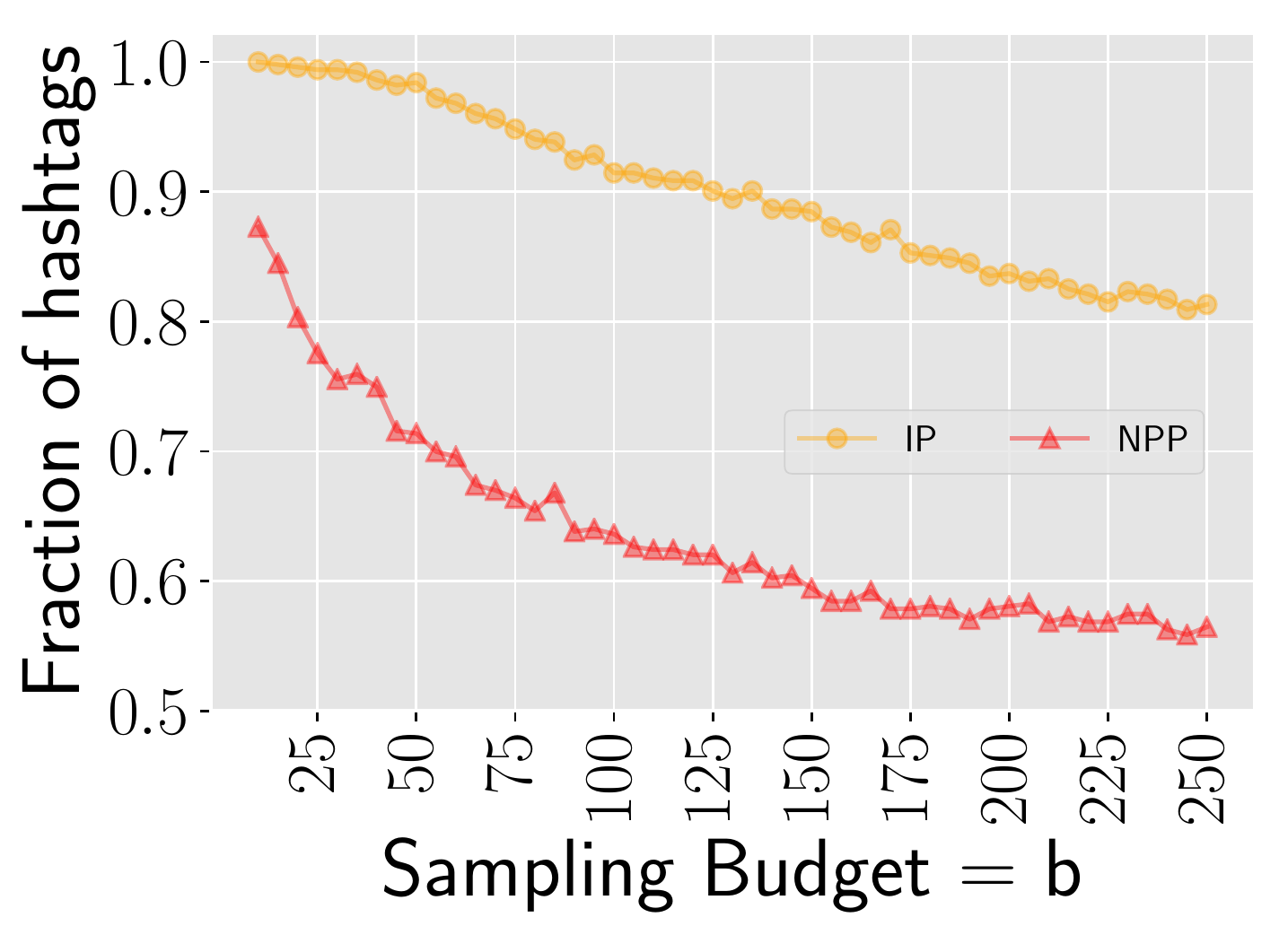}
        \caption{\label{fig:polling_better_ratio}}
    \end{subfigure}
    \caption{Comparison of polling algorithms for estimating the global prevalence of Twitter hashtags. Variation of (a) squared bias ( $\bias\{T\}^2$ ), (b) variance ( $\var\{T\}$ ) and (c) mean squared error ( $\bias\{T\}^2 + \var\{T\}$ ) of the polling estimate (\textit{IP}, \textit{NPP} and \textit{FPP} as $T$ - polling algorithm -) as a function of a hashtag's global prevalence $\mathbb{E}\{f(X)\}$. Each point represents a different hashtag and a fixed sampling budget $b = 25$.
    (d) Fraction of hashtags where the proposed FPP algorithm outperforms the other two in terms of mean squared error. The fraction for NPP approaches 0.5, and for IP approaches 0.8 as sampling budget $b$ increases. \label{fig:polling} These figures illustrate that the proposed FPP algorithm achieves a bias-variance trade-off by coupling perception polling with friendship paradox to reduce the mean squared error.}
\end{figure}

The FPP algorithm assumes that
every node has a non-zero in-degree and out-degree.
To evaluate the polling algorithm, we extract a subgraph of 5409 Twitter users from our dataset with the same properties. We use the polling algorithm to estimate the popularity of the 500 most frequent hashtags mentioned by users in this subgraph. 
We compare the performance of the proposed FPP algorithm on this induced subgraph to two alternative algorithms:
\begin{enumerate}
\item \textbf{Intent Polling - IP}: asks random users whether they used a hashtag (\textit{orange in Figure \ref{fig:polling}}).
\item \textbf{Node Perception Polling - NPP}: asks random users what fraction of their friends used the hashtag (\textit{red in Figure \ref{fig:polling}}).
\item \textbf{Follower Perception Polling - FPP}: asks random followers what fraction of their friends used the hashtag (\textit{green in Figure \ref{fig:polling}}).
\end{enumerate}

Node perception polling (NPP) differs from IP in terms of the questions asked: random nodes are asked about their perception in NPP, whereas they are asked about their attribute in IP. Follower perception polling (FPP) differs from NPP in terms of the sampling method: random followers are sampled (based on friendship paradox) in FPP whereas naive sampling of random nodes is used in NPP. Hence, comparing performance of IP with NPP will illustrate the effect of polling perceptions instead of attributes. Comparing performance of FPP with NPP will illustrate the effect of friendship paradox based perception polling in contrast to the naive sampling based perception polling.

Figure \ref{fig:bias_polling} shows the (empirical) squared bias of the three polling algorithms for a fixed sampling budget $b = 25$, which corresponds to querying $0.5\%$ of the nodes. As asserted in Theorem~\ref{th:bias_algFPP} (see Method), for each hashtag, the \textit{FPP} estimate is biased by an amount equal to $\Bglobal$ value for that hashtag. Hence, the \textit{IP}, which yields an unbiased estimate, outperforms \textit{NPP} and \textit{FPP} in terms of bias for most hashtags as illustrated in Figure~\ref{fig:bias_polling}. However, as illustrated in Figure \ref{fig:var_polling}, \textit{FPP} produces a smaller variance estimate compared to both \textit{IP} and \textit{NPP}. Hence, in terms of the \textit{Mean Squared Error} (which is defined as $\mse\{T\} = 
 \bias\{T\}^2 + \var\{T\}$ for an estimate $T$), FPP estimate is more accurate compared to both IP and NPP estimates for most hashtags as illustrated in Figure~\ref{fig:mse_polling}.
Increasing the sampling budget decreases performance gap between \textit{FPP} and the other two algorithms (Figure \ref{fig:polling_better_ratio}). However, even with $b = 250$ ($5\%$ of the nodes polled), \textit{FPP} outperforms \textit{IP} in more than $80\%$ of the cases, and it outperforms \textit{NPP} in more than $55\%$ of the cases.




\section{Discussion}
\label{sec:conclusion}
Social networks can have surprising, even counter-intuitive behaviors. For example, previous work has shown that the ``majority illusion'' may lead people to observe that the majority of their friends has some attribute, even when it is globally rare~\cite{Lerman2016majority}. The illusion is created by the friendship paradox, which can also bias the observations individuals make in directed networks in non-obvious ways. Our analysis identifies the conditions under which
%
%
friendship paradox can distort how popular some attribute or behavior (e.g., drinking, smoking, etc.) is perceived to be, making it appear several times more prevalent than it actually is. Specifically, the following two conditions amplify local perception bias: (1) positive correlation between an individual's attribute and popularity (number of followers in a directed network) and (2) positive correlation between the attributes of individuals and the attention of their followers. The first condition suggests that bias exists when popular people (i.e., those followed by many others and hence more visible) have the attribute, for example, engaging in risky behavior, having a specific political affiliation, or simply using a particular hashtag. Their influence is amplified when they are followed or seen by ``good listeners'', i.e., people who follow fewer others and thus are able to pay more attention to the influentials.

We validated these findings empirically using data from the Twitter social network. We measured perceptions of popularity of hashtags, i.e., words or phrases preceded by a `\#' sign that are frequently used to identify topics on Twitter. Such hashtags serve many important functions, from organizing content, to expressing opinions, to linking topics and people.  We measured a hashtag's global prevalence as the fraction of all people using it, and its perceived popularity as the fraction of friends using it.
Our analysis identified hashtags that appeared several times more popular  than they actually were, due to local perception bias. Such hashtags were associated with social movements, memes and current events.
Interestingly, as our data was collected in 2014, some of the most biased hashtags were \textit{\#icebucketchallenge} and \textit{\#alsicebucketchallenge}, the explosively popular Ice Bucket Challenge.
Perception bias could have potentially amplified their spread, as well as the spread of other costly behaviors that require social proof~\cite{cheng2018diffusion}. For example, the \#MeToo movement has grown into an international campaign to end sexual harassment and assault in the workplace by highlighting just how endemic the problem is. It spread through online social networks as women posted their own stories of harassment using the hashtag \textit{\#metoo}. Perception bias may have amplified the spread of such hashtags by making them appear more common and thus easier to use.


We also presented an algorithm that leverages friendship paradox in directed networks to efficiently (in a mean-squared error sense) estimate the true prevalence of an attribute. In essence, the idea behind the algorithm is that perceptions of random followers should have a smaller variance compared to the perceptions of random individuals. This is because random followers are more informed than random individuals (according to friendship paradox). It was shown that the variance of this algorithm is bounded by a function of the second largest eigenvalue of the degree-discounted bibliographic coupling matrix and the correlation between the out-degree and the attribute. Empirical results illustrate that the proposed algorithm outperforms other widely used polling algorithms.

Our work suggests that one way to mitigate perception bias is to alter the local network topology to allow more information to reach the low-attention users. {This opens up new research avenues on how link recommendation can alleviate perception bias.}
However, our empirical study has limitations, namely, the nature of the subsample of the network we studied.
Social networks are huge, 
necessitating analysis of subgraphs sampled from the entire network. However, by leaving out some nodes, data collection process itself may distort the properties of the sample. Specifically, since we observed only the outgoing links from the seed nodes, we do not have information about the followers of these nodes. Addressing the limitations of analysis imposed by sampling is an important research direction.
Despite this limitation, our work shows that friendship paradox can lead to surprising biases, especially in directed networks, and suggests potential strategies for mitigating these biases.

\section{Methods}

\label{sec:empirical}

\subsection{Data}
\label{subsecc:data}
The dataset used in this study was collected from Twitter in 2014. We started with a set of 100 users who were active discussing ballot initiatives during the 2012 California election and expanded this set by retrieving the accounts of the individuals they followed and reached a total of 5,599 users. We refer these individuals as \textit{seed users}. Next, we identified all friends of the seed users, collecting all directed links that start with one of the seed users. We then collected all posts made by the seed users and their friends---over 600K users in total---over the period June--November 2014. The posts include their activity i.e. tweets and retweets. These tweets mention more than 18M hashtags.  With this data-collection approach, seed users are fully observed (their activity and what they see in their social feeds), and their friends are only partially observed (only their activity).

Table \ref{tbl:statistics} reports properties of the Twitter dataset, considering only the \textit{seed users} and using the variables defined in Section~\ref{subsec:four_versions_of_DFP}. Note that the average degree $\bar{d}$ (where, $\bar{d} = \mathbb{E}\{\od(X)\} = \mathbb{E}\{\id(X)\}$) is relatively large at $123.55$. However, since the distribution of the in- and out-degree is highly heterogeneous, the variance of the in- and out-degrees is relatively large (two orders of magnitude compared to $\bar{d}$).
The covariance between the in- and out-degrees of nodes is also relatively large with a correlation coefficient $\rho\{\id(X),\od(X)\} = \cov \{\id(X),\od(X)\}/\sqrt{\var\{\od(X)\} \var\{\id(X)\}} = 0.52$.

\begin{table}
  \centering
  \caption{Properties of the Twitter subgraph (Sec. \ref{subsecc:data})}\label{tbl:statistics}
\resizebox{0.95\columnwidth}{!}{%
\begin{tabular}{|c|c|c|}
  \hline
  \multicolumn{3}{|c|}{\emph{Properties of nodes} } \\
  \hline
  avg. degree &  $\avgdegree=\mathbb{E}\{d_{i}(X)\}$ & $123.55$ \\
  \hline
variance of out-degree &  $\var\{\od(X)\}$ & $30096.16$ \\
  \hline
variance of in-degree &  $\var\{\id(X)\}$ & $24338.66$ \\
  \hline
covariance   &   $\cov\{\id(X),\od(X)\}$ & $14226.32$ \\
  \hline
  \multicolumn{3}{|c|}{\emph{Properties of friends and followers} } \\
  \hline
friend's avg. out-degree &  $\mathbb{E}\{\od(Y)\}$  & $367.14$ \\
  \hline
friend's  avg. in-degree &  $\mathbb{E}\{\id(Y)\}$  & $238.68$  \\
  \hline
follower's avg. in-degree &  $\mathbb{E}\{\id(Z)\}$  & $320.54$  \\
  \hline
follower's avg. out-degree &  $\mathbb{E}\{\od(Z)\}$ & $238.68$  \\
  \hline
\end{tabular}}
\end{table}
Due to the relatively large variance (compared to $\bar{d}$) of the in- and out-degree distributions, the expected out-degree of a random friend ($\mathbb{E}\{\od(Y)\}$) and the expected in-degree of a random follower ($\mathbb{E}\{\id(Z)\}$) are larger than the average degree $\avgdegree$ as stated in Theorem \ref{th:friendship_paradox_any_network}. Note also that, due to positive covariance between the in- and out-degrees of nodes, the expected in-degree of a random friend ($\mathbb{E}\{\id(Y)\}$) and the expected out-degree of a random follower ($\mathbb{E}\{\od(Z)\}$) are also larger than $\avgdegree$, as stated in Theorem \ref{th:friendship_paradox_positivelyCorrelated_inout_degree_network}.

\subsection{Friendship Paradox-based Polling: Performance Analysis}
\label{subsec:unbiased_estimate}
The accuracy of a poll depends on the method of sampling respondents and the question asked of them.  For example, in the case of estimating an election outcome, asking people \textit{``Who do you think will win?"} (expectation polling) is better than \textit{``Who will you vote for?"} (intent polling)~\cite{rothschild2011forecasting}. This is because in expectation polling, an individual names the candidate more popular among her friends, thus summarizing a number of individuals in the social network, rather that provide her own voting intention. 
Our polling algorithm is motivated by \cite{dasgupta2012social,rothschild2011forecasting,nettasinghe2018your}, which show that polling methods asking individuals to summarize information in their neighborhood outperform polling methods that ask only about the attribute of each individual. {\cite{dasgupta2012social} studied the polling problem analytically in the context of an undirected network and, proposed a method to obtain an unbiased estimate of the global prevalence with bounds on its variance. }The analysis of Algorithm 1 
for directed graphs is motivated by these results in \cite{dasgupta2012social} for undirected social networks.  \cite{nettasinghe2018your} proposed to ask the simple question ``What fraction of your neighbors have the attribute 1?" (neighborhood expectation polling) from randomly sampled neighbors (instead of random nodes) on undirected social networks. In this case, sampled individuals will provide the average opinion among their neighbors. Further, since random friends have more friends than random individuals (by the friendship paradox for undirected graphs), this approach would yield an estimate with a smaller variance than asking it from random nodes. Motivated by these works, Algorithm 1 exploits the friendship paradox on directed networks to obtain a statistically efficient estimate of the global prevalence of an attribute using biased perceptions of random followers.

\paragraph{Analysis of the FPP Algorithm}
Recall that in order to reduce the variance, the FPP algorithm polls perceptions $q_f(Z)$ of random followers $Z$ instead of attributes $f(X)$ of random individuals $X$. However, it is not guaranteed that the estimate $\FPPest$ will be unbiased. The following result shows that the bias of the FPP algorithm is the same as the global perception bias $\Bglobal$.
\begin{theorem}
	\label{th:bias_algFPP}
	The bias of the estimate $\FPPest$ computed in Algorithm~1 is equal to the global perception bias $\Bglobal$ i.e.
	\begin{align}
	\bias(\FPPest) &= \mathbb{E}\{\FPPest\} - \mathbb{E}\{f(X)\}\\
	&= \Bglobal
	\end{align}
\end{theorem}

Hence, the same factors (specified in Eq. (\ref{eq:global_perception_bias_Y})) that increase (decrease) the global perception bias will increase (decrease) the bias of the estimate $\FPPest$ produced by the FPP algorithm. The aim of the FPP algorithm is to compensate for the bias $\Bglobal$ of the algorithm with a reduced variance and thereby achieve a smaller mean squared error. Also, we highlight that the Algorithm 1 can be modified to generate an unbiased estimate by replacing (\ref{eq:NEP_estimate}) with
	\begin{align}
	\label{eq:modified_unabiased_FPPest}
	\hat{f}_{FPP}^{\text{Unbiased}} = \frac{1}{b} \sum_{ v \in S} \frac{1}{Np_v}\sum_{ u \in {Fr}(v)}\frac{f(u)}{\od(u)}.
	\end{align}  
{The unbiased estimate $\hat{f}_{FPP}^{\text{Unbiased}} $ is based on the concept of \textit{social sampling} proposed in \cite{dasgupta2012social} for undirected social networks where, queried individuals provide a weighted value of their friends' attributes in a manner that results in an unbiased estimate.} This estimate is useful in contexts where unbiasedness is preferred over mean-squared error to assess the performance of the estimate. However, this does not result in an intuitive and easily implementable algorithm similar to Algorithm~1 since the modified estimate $\hat{f}_{FPP}^{\text{Unbiased}}$ involves each sampled individual calculating a weighted average of the attributes of her neighbors.

Before analyzing the variance of estimate $\hat{f}$ produced by the Algorithm 1, we digress briefly to review the \emph{bibliographic coupling matrix}. Bibliographic coupling originated from the analysis of citation networks \cite{kessler1963bibliographic}, and is used to symmetrize a directed graph by transform it into an undirected graph for purposes of clustering, etc. The bibliographic coupling matrix $B$ of a directed graph with adjacency matrix $A$ is defined as $B = AA^T$. Hence, the weight of the link between nodes $i, j$ in the new undirected graph is $B(i,j) = \sum_{v \in V}A(i,v)A(j, v)$ which corresponds to the number of mutual followers of $i$ and $j$. 
Hence, the weight of the link between two nodes $i$ and $j$ in $B$ is the number of individuals who follow both of these nodes.\footnote{In a citation network where the nodes correspond to papers, the entry $(i,j)$ of the bibliographic coupling matrix $B$ gives the number of papers that are cited by both $i$ and $j$ from which the name Bibliographic coupling matrix is derived. Bibliographic coupling matrix is also related to the HITS algorithm \cite{kleinberg1999authoritative} used for link analysis \cite{newman2010networks,liu2007web}.} This conveys the similarity of $i, j$ in terms of the number of mutual followers. However, when determining the similarity of two nodes $i, j$ using $B$, a mutual follower with a large number of 
friends (a likely scenario), is weighted the same as a mutual follower with a small number of 
friends (a rarer scenario). Hence, the latter type of mutual follower should be given more weight compared to the former type when evaluating the similarity of two nodes. Similarly, the number of followers of $i$ and $j$ should also be taken into consideration when assessing their similarity. Based on these observations, \cite{satuluri2011symmetrizations} proposed the degree-discounted bibliographic coupling matrix
\begin{equation}
B_d = D_o^{-1/2}AD_i^{-1}A^TD_o^{-1/2}
\end{equation}
where $D_o$ and $D_i$ are the $N \times N$ dimensional diagonal matrices with $D_o(i,i) = \od(i)$ and $D_i(i,i) = \id(i)$, respectively. The $(i, j)$ element of $B_d$ is
\begin{equation}
B_d(i,j) = \frac{1}{\sqrt{\od(i)\od(j)}}\sum_{ k \in V}\frac{A(i,k)A(j,k)}{{\id(k)}},
\end{equation}
which discounts the contributions of the nodes $i, j$ by their out-degrees (the number of followers) and each mutual follower $k$ by her in-degree (number of friends). Please see \cite{satuluri2011symmetrizations,malliaros2013clustering} for more details on the degree-discounted bibliographic coupling.

Returning to the analysis of the estimate $\FPPest$ of the Algorithm 1, the following result gives an upper bound on the variance of this estimate under certain conditions on the structure of the network.

\begin{theorem}
	\label{th:variance_algFPP}
	Consider the estimate $\FPPest$ generated by Algorithm 1 for a graph $G = (V,E)$ with labels $f:V\rightarrow \{0,1\}$. If the degree-discounted bibliographic coupling matrix $B_d$ is connected, non-bipartite, then
	\begin{align}
	\var(\FPPest) &= \frac{ f^TD_o^{1/2}}{bM}\bigg(		D_o^{-1/2}A	D_i^{-1}A^{T}D_o^{-1/2}	 - \frac{D_o^{1/2}\mathds{1}\mathds{1}^{T}D_o^{1/2}}{M}	\bigg)D_o^{1/2}f\\
	&\leq \frac{1}{bM} \lambda_2 || D_o^{1/2}f||^2
	\end{align} where, $M = \sum_{ v \in V}\id(v)$, $\lambda_2$ is the second largest eigenvalue of $B_d$, $f$ is the $N\times1$ dimensional vector of binary attributes.
\end{theorem}


Theorem \ref{th:variance_algFPP} shows that the variance of the friendship paradox based polling Algorithm 1 depends on the correlation between the out-degrees and attributes $|| D_o^{1/2}f||^2 $ and the structure of the graph via second largest eigenvalue $\lambda_2$ of the matrix $B_d$. Specifically, a smaller $\lambda_2$ implies that the bibliographic coupling network has a good expansion (i.e. absence of bottlenecks) \cite{estrada2006network}. Hence, if the nodes in the network $G =(V, E)$ cannot be clustered into distinct groups based on their mutual followers (i.e. bibliographic similarity) then, the variance of the algorithm will be smaller (due to smaller $\lambda_2$).

\appendix
\section*{Appendix}
\renewcommand{\thesection}{A\arabic{section}}  
\renewcommand{\theequation}{A\arabic{equation}}   


\section{Proof of Theorem 1}
\label{appendix:th_FP_any_network}
	\begin{theorem*}
		Let $G = (V,E)$ be a directed network. Then,
		\begin{compactenum}
			\item random friend $Y$ has more followers than a random node $X$, on average; i.e.,
			\begin{equation}
\label{eq:dfpout}
				\mathbb{E}\{\od(Y)\} - \avgdegree = \frac{\var\{\od(X)\}}{\avgdegree} \geq 0.
			\end{equation}
			
			\item random follower $Z$ has more friends than a random node $X$, on average; i.e.,
			\begin{equation}
\label{eq:dfpin}
				\mathbb{E}\{\id(Z)\} - \avgdegree = \frac{\var\{\id(X)\}}{\avgdegree} \geq 0.
			\end{equation}
		\end{compactenum}
	\end{theorem*}
\begin{proof}

\noindent
Part 1:

	$\hspace{0cm} \mathbb{E}\{\od(Y)\} - \mathbb{E}\{\od(X)\} = \sum_{v\in V} \od(v)\mathbb{P}(Y = v)- \sum_{v\in V}\frac{\od(v)}{N}$
	\begin{align}
	&\hspace{0cm}=\sum_{v\in V}\od(v)\frac{\od(v)}{\sum_{v'\in V}\od(v')} - \frac{\sum_{v\in V}\od(v)}{N}\\
	&\hspace{0cm}=  \frac{\frac{\sum_{v\in V}\od(v)^2}{N} - \bigg(\frac{\sum_{v\in V}\od(v)}{N}\bigg)^2}{\frac{\sum_{v'\in V}\od(v')}{N}}\\
	&\hspace{0cm}=\frac{\mathbb{E}\{\od(X)^2\} - \mathbb{E}\{\od(X)\}^2}{\mathbb{E}\{\od(X)\}} = \frac{\var\{\od(X)\}}{\avgdegree} \geq 0
	\end{align}
\noindent
Proof of part 2 follows using similar arguments.
\end{proof}

\section{Proof of Theorem 2}
\label{appendix:th_FP_postivelyCorrelate_inout_degree}

	\begin{theorem*}
		Let $G = (V,E)$ be a directed network where in-degree $\id(X)$ and out-degree $\od(X)$ of a random node $X$ are positively correlated. Then,
		\begin{compactenum}
			\item random friend $Y$ has more friends than a random node $X$ does, on average; i.e.,
			\begin{equation}
				\mathbb{E}\{\id(Y)\} - \avgdegree = \frac{\cov\{\id(X),\od(X)\}}{\avgdegree} \geq 0.
			\end{equation}
			
			\item random follower $Z$ has more followers than a random node $X$ does, on average; i.e.,
			\begin{equation}
				\mathbb{E}\{\od(Z)\} - \avgdegree = \frac{\cov\{\id(X),\od(X)\}}{\avgdegree} \geq 0.
			\end{equation}
		\end{compactenum}
	\end{theorem*}

\begin{proof}
\noindent
Part 1:

	$\mathbb{E}\{\id(Y)\} - \mathbb{E}\{\id(X)\} = \sum_{v\in V} \id(v)\mathbb{P}(Y = v)- \sum_{v\in V}\frac{\id(v)}{N}$
	\begin{align}
	&\hspace{0.5cm}=\sum_{v\in V}\id(v)\frac{\od(v)}{\sum_{v'\in V}\od(v')} - \frac{\sum_{v\in V}\id(v)}{N}\\
	&\hspace{0.5cm}=  \frac{\frac{\sum_{v\in V}\id(v)\od(v)}{N} - \bigg(\frac{\sum_{v\in V}\id(v)}{N}\bigg)\bigg(\frac{\sum_{v'\in V}\od(v')}{N}\bigg)}{\frac{\sum_{v'\in V}\od(v')}{N}}\\
	&\hspace{0.5cm}=\frac{\mathbb{E}(\id(X)\od(X)) - \mathbb{E}\{\id(X)\}\mathbb{E}\{\od(X)\}}{\mathbb{E}\{\od(X)\}} = \frac{\cov\{\id(X),\od(X)\}}{\avgdegree}
	\end{align}
Hence, positive correlation ($\cov\{\id(X),\od(X)\} > 0$) between in-degree $\id(X)$ and out-degree $\od(X)$ of a random individual $X$ implies that
$\mathbb{E}\{\id(Y)\} > \mathbb{E}\{\id(X)\}$.

Proof of part 2 follows using similar arguments.
\end{proof}


\section{Derivation of $\Blocal$}
\label{appendix:th_local_PB}
Let $Y'$ denote a uniformly sampled friend of a random node $X$. Further, let $A_{uv}$ denote the element $(u,v)$ of the adjacency matrix of network: $A_{uv} = 1$ if there is a link pointing from $u$ to $v$ and $A_{uv} = 0$ otherwise. Then, by definition of the function $q_f$ in Section 2 of the main text, 
\begin{align}
q_f(X) = \frac{\sum_{ U \in {Fr}(X)}f(U)}{\id (X)} = \mathbb{E}\{f(Y')\vert X\}
\end{align}
Therefore,
\begin{equation}
\mathbb{E}\{q_f(X)\}= \frac{1}{N}\sum_{v\in V}\Bigg\{  \frac{\sum_{ u \in {Fr}(v)}f(u)}{\id(v)}\Bigg\} = \frac{1}{N}\sum_{v\in V}\Bigg\{  \sum_{u \in v}\frac{f(u)}{\id(v)}A_{uv}\Bigg\}
\end{equation}
\begin{align}
&= \frac{{\sum_{u,v \in V}A_{uv}}}{N}\sum_{v\in V}\Bigg\{  \sum_{u \in V}\frac{f(u)}{\id(v)}\frac{A_{uv}}{\sum_{u,v \in V}A_{uv}}\Bigg\}\\
&=\avgdegree\times\mathbb{E}\bigg\{ \frac{f(U)}{\id(V)}\bigg\vert (U,V) \sim \uniform(E)   \bigg\}
\end{align}
which proves the first statement.

Next, assume, $f(U)$ and $\mathcal{A}(V)$ (where, $(U,V)$ is a random link) are positively correlated ($\cov\{f(U), \mathcal{A}(V)\} \geq 0$). Then,
\begin{align}
\mathbb{E}\{q_f(X)\}&=\avgdegree\mathbb{E}\Big\{ {f(U)}{\mathcal{A}(V)}\Big\vert (U,V) \sim \uniform(E)   \Big\}\\
&\geq \avgdegree \mathbb{E}\{f(U) \vert (U,V) \sim \uniform(E)   \}\\
&\hspace{1cm}\times \mathbb{E}\{\mathcal{A}(V) \vert (U,V) \sim \uniform(E)   \}\nonumber\\
&= \mathbb{E}\{f(Y)\}
\end{align}

Therefore, $\cov\{f(U), \mathcal{A}(V)\} \geq 0$ (condition (14) in the main text
) implies $\mathbb{E}\{q_f(X)\} \geq \mathbb{E}\{f(Y)\}$.
Also, from Theorem 2,
${\cov\{f(X),\od(X)\} \geq 0}$ (condition  (13) in the main text)
implies $\mathbb{E}\{f(Y)\} \geq \mathbb{E}\{f(X)\}$.
The proof follows.

\section{Proof of Theorem 3}
\label{appendix:th_bias_algFPP}

\begin{theorem*}
	The bias of the estimate $\FPPest$ computed in Algorithm~1 is equal to the global perception bias $\Bglobal$ i.e.
	\begin{align}
	\bias(\FPPest) &= \mathbb{E}\{\FPPest\} - \mathbb{E}\{f(X)\}\\
	&= \Bglobal
	\end{align}
\end{theorem*}

\begin{proof}

Let $e_v$ denote the $n\times 1$ dimensional unit vector with $1$ at the $v^{th}$ element and zeros elsewhere. Then,
\begin{align}
\label{eq:q_f_v}
q_f(v) =e_v^TD_i^{-1}A^Tf
\end{align}
and let $M = \sum_{v \in V} {\id(v)}$. With $Z$ defined in Equation (3) of the main text, 
\begin{align}
\mathbb{E}\{\FPPest\} &= \mathbb{E}\{q_f(Z)\} = \sum_{v \in V} \frac{\id(v)}{M}q_f(v)\\
&=  \sum_{v \in V} \frac{\id(v)}{M}\bigg(e_v^TD_i^{-1}A^Tf\bigg)= \frac{1}{M}\mathds{1}^TD_iD_i^{-1}A^Tf \\
&= \frac{1}{M}\mathds{1}^TA^Tf \label{eq:exp_q_Z}\\
&= \sum_{v\in V}f(v)\frac{\od(v)}{M} = \mathbb{E}\{f(Y)\}
\end{align}
Therefore,
\begin{align}
\bias\{\FPPest\} &= \mathbb{E}\{\FPPest\} - \mathbb{E}\{f(X)\}\\
&= \mathbb{E}\{f(Y)\} - \mathbb{E}\{f(X)\} = \Bglobal
\end{align}
\end{proof}

\section{Proof of Theorem 4}
\label{appendix:th_variance_algFPP}

\begin{theorem*}
	Consider the estimate $\FPPest$ generated by Algorithm 1 for a graph $G = (V,E)$ with labels $f:V\rightarrow \{0,1\}$. If the degree-discounted bibliographic coupling matrix $B_d$ is connected, non-bipartite, then
	\begin{align}
	\var(\FPPest) &= \frac{ f^TD_o^{1/2}}{bM}\bigg(		D_o^{-1/2}A	D_i^{-1}A^{T}D_o^{-1/2}	 - \frac{D_o^{1/2}\mathds{1}\mathds{1}^{T}D_o^{1/2}}{M}	\bigg)D_o^{1/2}f\\
	&\leq \frac{1}{bM} \lambda_2 || D_o^{1/2}f||^2
	\end{align} where, $M = \sum_{ v \in V}\id(v)$, $\lambda_2$ is the second largest eigenvalue of $B_d$, $f$ is the $N\times1$ dimensional vector of binary attributes.
\end{theorem*}

\begin{proof}
Since $\FPPest$ is the average of the perceptions of $b$ independently sampled random followers,
\begin{align}
\var(\FPPest) &= \frac{1}{b}\var(q_f(Z)) =\frac{1}{b}\Big(\mathbb{E}\{q^2_f(Z)\} - \mathbb{E}\{q_f(Z)\}^2\Big) \label{eq:var_qf_z}
\end{align}
where, $Z$ is a random follower.
Consider $\mathbb{E}\{q^2_f(Z)\}$.
\begin{align}
\mathbb{E}\{q^2_f(Z)\} &= \sum_{ v \in V}\frac{\id(v)}{M}q^2_f(v) = \sum_{ v \in V} \frac{\id(v)}{M}f^TAD_{i}^{-1}e_ve_v^TD_i^{-1}A^Tf\\
&\hspace{0.5cm}\text{(by substituting for $q_f(v)$ from (\ref{eq:q_f_v}))} \nonumber\\
&=\frac{1}{M}\bigg(	f^TAD_i^{-1}\Big(\sum_{ v \in V}\id(v)e_ve_v^T\Big)D_i^{-1}A^Tf\		\bigg)\\
&=\frac{1}{M}f^TAD_{i}^{-1}A^Tf \label{eq:exp_qf_squared_z}
\end{align}
Hence,
\begin{align}
\var(q_f(Z)) &= \mathbb{E}\{q^2_f(Z)\} - \mathbb{E}\{q_f(Z)\}^2\\
&= \frac{1}{M}f^TAD_{i}^{-1}A^Tf - \frac{1}{M^2}f^TA\mathds{1}\mathds{1}^TA^Tf\\
&\hspace{0.5cm}\text{(by substituting from (\ref{eq:exp_q_Z}) and (\ref{eq:exp_qf_squared_z}))}\nonumber\\
&=\frac{1}{M}f^T\bigg(	AD_{i}^{-1}A^T - 	\frac{1}{M}{A\mathds{1}\mathds{1}^TA^T}	\bigg)f\\
&= \frac{ f^TD_o^{1/2}}{M}\bigg(		D_o^{-1/2}A	D_i^{-1}A^{T}D_o^{-1/2}	 - \frac{D_o^{1/2}\mathds{1}\mathds{1}^{T}D_o^{1/2}}{M}	\bigg)D_o^{1/2}f\\
&\leq\frac{|| D_o^{1/2}f\vert\vert^2}{M}	\bigg|\bigg| 	D_o^{-1/2}A	D_i^{-1}A^{T}D_o^{-1/2}	 - \frac{D_o^{1/2}\mathds{1}\mathds{1}^{T}D_o^{1/2}}{M}	\bigg|\bigg| \label{eq:cauchy_schwarz_spectral_norm}
\end{align}
where, for a matrix $A$, $||A||$ denotes the spectral norm (largest singular value) and (\ref{eq:cauchy_schwarz_spectral_norm}) is obtained by applying the Cauchy-Schwarz inequality.

Note that the degree-discounted bibliographic coupling-matrix, 
$$B_d = D_o^{-1/2}A	D_i^{-1}A^{T}D_o^{-1/2} = \big(D_o^{-1/2}A	D_i^{-1/2}\big){\big(D_o^{-1/2}A	D_i^{-1/2}\big)}^T$$ 
is a symmetric, positive semi-definite matrix. Hence, all eigenvalues are non-negative. Further, $\frac{D_o^{1/2}\mathds{1}}{\sqrt{M}}$ is the eigenvector with all non-negative elements and corresponds to eigenvalue $1$. Hence,
$$	\big|\big| 	D_o^{-1/2}A	D_i^{-1}A^{T}D_o^{-1/2}	 - \frac{D_o^{1/2}\mathds{1}\mathds{1}^{T}D_o^{1/2}}{M}	\big|\big| = \lambda_2$$ 
where, $\lambda_2$ is the second largest eigenvalue of
$B_d$. Then, the result follows from (Equation \ref{eq:var_qf_z}).
\end{proof}
\newpage
\section{Appendix Figures}
\begin{figure}[hb!]
    \includegraphics[width=0.84\columnwidth]{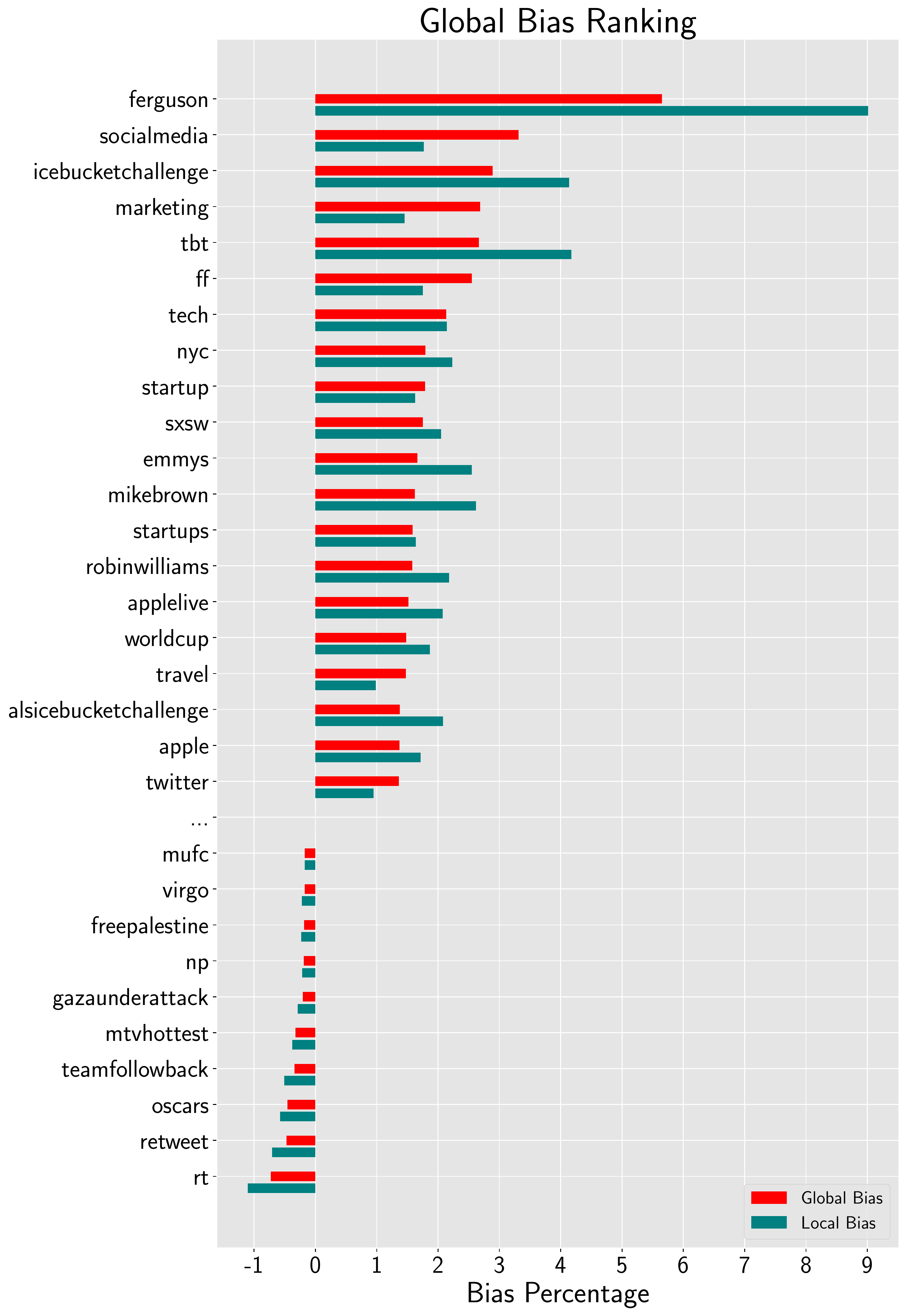}
    \caption{The ranking of popular Twitter hashtags based on \textit{Global bias}. Top-$20$ and bottom-$10$ are included in the ranking. The bars compare Global bias ($\Bglobal$) and Local Bias ($\Blocal$). There are $94$ hashtags among $1153$ with opposite sign of local bias and global bias, although both bias values for these hashtags are close to zero.
    Among the remaining hashtags, $661$ ($62\%$) have larger local bias than global bias, and $398$ ($38\%$) have larger global bias than local bias. \label{fig:global_bias_ranking}}
\end{figure}
\newpage
\bibliographystyle{unsrt}
\bibliography{references}

\begin{thebibliography}{10}

\bibitem{Miller1994collective}
Dale~T. Miller and Deborah~A. Prentice.
\newblock Collective errors and errors about the collective.
\newblock {\em Personality and Social Psychology Bulletin}, 20(5):541--550,
  October 1994.

\bibitem{Baer91}
J.~S. Baer, A.~Stacy, and M.~Larimer.
\newblock {Biases in the perception of drinking norms among college students.}
\newblock {\em Journal of studies on alcohol}, 52(6):580--586, November 1991.

\bibitem{Prentice1993pluralistic}
Deborah~A. Prentice and Dale~T. Miller.
\newblock Pluralistic ignorance and alcohol use on campus: Some consequences of
  misperceiving the social norm.
\newblock {\em Journal of Personality and Social Psychology}, 64(2):243--256,
  1993.

\bibitem{Kitts2003egocentric}
James~A. Kitts.
\newblock Egocentric bias or information management? selective disclosure and
  the social roots of norm misperception.
\newblock {\em Social Psychology Quarterly}, 66(3):222--237, 2003.

\bibitem{berkowitz2005overview}
Alan~D Berkowitz.
\newblock {An overview of the social norms approach}.
\newblock {\em Changing the culture of college drinking: A socially situated
  health communication campaign}, pages 193--214, 2005.

\bibitem{Feld91}
Scott~L. Feld.
\newblock {Why Your Friends Have More Friends Than You Do}.
\newblock {\em American Journal of Sociology}, 96(6):1464--1477, May 1991.

\bibitem{Bollen11}
Johan Bollen, Bruno Gon\c{c}alves, Guangchen Ruan, and Huina Mao.
\newblock {Happiness Is Assortative in Online Social Networks}.
\newblock {\em Artificial Life}, 17(3):237--251, May 2011.

\bibitem{Benevenuto2016}
Fabr{\'i}cio Benevenuto, Alberto H.~F. Laender, and Bruno~L. Alves.
\newblock The h-index paradox: your coauthors have a higher h-index than you
  do.
\newblock {\em Scientometrics}, 106(1):469--474, Jan 2016.

\bibitem{Lerman2016majority}
Kristina Lerman, Xiaoran Yan, and Xin-Zeng Wu.
\newblock The" majority illusion" in social networks.
\newblock {\em PloS one}, 11(2):e0147617, 2016.

\bibitem{Eom14}
Young-Ho Eom and Hang-Hyun Jo.
\newblock {Generalized friendship paradox in complex networks: The case of
  scientific collaboration}.
\newblock {\em Scientific Reports}, 4, April 2014.

\bibitem{abel2016social}
Jessica~P Abel, Cheryl~L Buff, and Sarah~A Burr.
\newblock Social media and the fear of missing out: Scale development and
  assessment.
\newblock {\em Journal of Business \& Economics Research (Online)}, 14(1):33,
  2016.

\bibitem{Hodas13icwsm}
Nathan Hodas, Farshad Kooti, and Kristina Lerman.
\newblock {Friendship Paradox Redux: Your Friends Are More Interesting Than
  You}.
\newblock In {\em {Proc. 7th Int. AAAI Conf. on Weblogs And Social Media}},
  2013.

\bibitem{Kooti14icwsm}
Farshad Kooti, Nathan~O. Hodas, and Kristina Lerman.
\newblock {Network Weirdness: Exploring the Origins of Network Paradoxes}.
\newblock In {\em {International Conference on Weblogs and Social Media
  (ICWSM)}}, March 2014.

\bibitem{higham2018centrality}
Desmond~J Higham.
\newblock Centrality-friendship paradoxes: When our friends are more important
  than us.
\newblock {\em arXiv preprint arXiv:1807.01496}, 2018.

\bibitem{Eom14_2}
Young-Ho Eom and Hang-Hyun Jo.
\newblock Generalized friendship paradox in networks with tunable
  degree--attribute correlation.
\newblock {\em Physical Review E.}, 90(2):022809, July 2014.

\bibitem{GomezRodriguez13}
Manuel~Gomez Rodriguez, Krishna Gummadi, and Bernhard Schoelkopf.
\newblock Quantifying information overload in social media and its impact on
  social contagions.
\newblock In {\em Eighth International AAAI Conference on Weblogs and Social
  Media}, 2014.

\bibitem{Hodas12socialcom}
Nathan~O. Hodas and Kristina Lerman.
\newblock How limited visibility and divided attention constrain social
  contagion.
\newblock In {\em ASE/IEEE International Conference on Social Computing}, 2012.

\bibitem{cheng2018diffusion}
Justin Cheng, Jon Kleinberg, Jure Leskovec, David Liben-Nowell, Bogdan State,
  Karthik Subbian, and Lada Adamic.
\newblock Do diffusion protocols govern cascade growth?
\newblock In {\em Proceddings of the International Conference on the Web and
  Social Media}, 2018.

\bibitem{rothschild2011forecasting}
David~M Rothschild and Justin Wolfers.
\newblock Forecasting elections: Voter intentions versus expectations.
\newblock 2011.

\bibitem{dasgupta2012social}
Anirban Dasgupta, Ravi Kumar, and D~Sivakumar.
\newblock Social sampling.
\newblock In {\em Proceedings of the 18th ACM SIGKDD international conference
  on Knowledge discovery and data mining}, pages 235--243. ACM, 2012.

\bibitem{nettasinghe2018your}
Buddhika Nettasinghe and Vikram Krishnamurthy.
\newblock What do your friends think? efficient polling methods for networks
  using friendship paradox.
\newblock {\em arXiv preprint arXiv:1802.06505}, 2018.

\bibitem{kessler1963bibliographic}
Maxwell~Mirton Kessler.
\newblock Bibliographic coupling between scientific papers.
\newblock {\em American documentation}, 14(1):10--25, 1963.

\bibitem{kleinberg1999authoritative}
Jon~M Kleinberg.
\newblock Authoritative sources in a hyperlinked environment.
\newblock {\em Journal of the ACM (JACM)}, 46(5):604--632, 1999.

\bibitem{newman2010networks}
M.~Newman.
\newblock {\em Networks: An Introduction}.
\newblock OUP Oxford, 2010.

\bibitem{liu2007web}
Bing Liu.
\newblock {\em Web data mining: exploring hyperlinks, contents, and usage
  data}.
\newblock Springer Science \& Business Media, 2007.

\bibitem{satuluri2011symmetrizations}
Venu Satuluri and Srinivasan Parthasarathy.
\newblock Symmetrizations for clustering directed graphs.
\newblock In {\em Proceedings of the 14th International Conference on Extending
  Database Technology}, pages 343--354. ACM, 2011.

\bibitem{malliaros2013clustering}
Fragkiskos~D Malliaros and Michalis Vazirgiannis.
\newblock Clustering and community detection in directed networks: A survey.
\newblock {\em Physics Reports}, 533(4):95--142, 2013.

\bibitem{estrada2006network}
Ernesto Estrada.
\newblock Network robustness to targeted attacks. the interplay of
  expansibility and degree distribution.
\newblock {\em The European Physical Journal B-Condensed Matter and Complex
  Systems}, 52(4):563--574, 2006.

\end{thebibliography}

\end{document}